\documentclass[aps,pra,twocolumn,superscriptaddress]{revtex4-1}
\usepackage[colorlinks=true, allcolors=blue]{hyperref}
\usepackage{tikz}
\usepackage{amsmath,nccmath}
\usepackage{amssymb}
\usepackage{graphicx}
\usepackage{physics}
\usepackage{makecell}
\usepackage{dsfont}
\usepackage{float}
\usepackage[T1]{fontenc}
\usepackage{booktabs}
\usepackage{amsthm}
\usepackage{mathrsfs}
\usepackage[title]{appendix}

\newcommand{\idd}
{\mathds{1}}
\usepackage{lipsum, babel}
\usetikzlibrary{positioning}
\usetikzlibrary{calc}
\usepackage[normalem]{ulem}
\usepackage{lipsum}
\newtheorem{thm}{Theorem}
\newtheorem{cnj}{Conjecture}
\newtheorem{lem}{Lemma}

\newtheorem{defi}{Definition}
\usepackage{nicefrac}

\usepackage{enumitem}
\setlist{nolistsep}

\begin{document}
\title{Translating Bell Non-Locality to Prepare-and-Measure Scenarios\\ under Dimensional Constraints}
\author{Matilde Baroni}
\email{matilde.baroni@lip6.fr}
\affiliation{Sorbonne Université, CNRS, LIP6, F-75005 Paris, France}
\author{Eleni Diamanti}
\affiliation{Sorbonne Université, CNRS, LIP6, F-75005 Paris, France}
\author{Damian Markham}
\affiliation{Sorbonne Université, CNRS, LIP6, F-75005 Paris, France}
\author{Ivan \v{S}upi\'{c}}
\affiliation{Sorbonne Université, CNRS, LIP6, F-75005 Paris, France}
\affiliation{ Universit\'e Grenoble Alpes, CNRS, Grenoble INP, LIG, 38000 Grenoble, France}

\begin{abstract}
    Understanding the connections between different quantum information protocols has been proven fruitful for both theoretical insights and experimental applications.
    In this work, we explore the relationship between non-local and prepare-and-measure scenarios, proposing a systematic way to translate bipartite Bell inequalities into dimensionally-bounded prepare-and-measure tasks.
    We identify sufficient conditions under which the translation preserves the quantum bound and self-testing properties, enabling a wide range of certification protocols originally developed for the non-local setting to be adapted to the sequential framework of prepare-and-measure with a dimensional bound. While the dimensionality bound is not device-independent, it still is a practical and experimentally reasonable assumption in many cases of interest.
    In some instances, we find new experimentally-friendly certification protocols. In others, we demonstrate equivalences with already known prepare-and-measure protocols, where self-testing results were previously established using alternative mathematical methods.
    Our results unify different quantum correlation frameworks, and contribute to the ongoing research effort of studying the interplay between parallel and sequential protocols.
\end{abstract}

\maketitle

\section{Introduction}

Quantum systems harbour the potential to significantly enhance various facets of information processing. This is evident in quantum key distribution (QKD), where the security of the protocol can be information-theoretic, even in the absence of trust in the devices employed for key generation~\cite{Vazirani_2014}. In fact, the main directions for the development of quantum communication technologies were laid down in two early seminal contributions to QKD, the BB84 protocol~\cite{bb84-a} and the Ekert protocol~\cite{Ekert}. 

Ekert's protocol explicitly capitalizes on Bell non-locality, highlighting the impossibility of simulating certain measurement correlations through locally causal models. Operationally, Bell non-locality, as articulated by Bell in 1964 ~\cite{Bell64}, is evidenced by the violation of Bell inequalities. Since non-locality is manifested in the correlations between measurement outcomes, the observed non-classicality is independent of the specific characteristics of the physical systems involved. This conceptual framework gave rise to the device-independent (DI) paradigm within quantum information processing~\cite{MayersYao,colbeck,pironio2010random}. In the context of DI QKD, the security of the generated key can be assured even in the presence of noisy devices or, remarkably, when there is no trust in the specifications of the devices used.
The DI paradigm has profoundly influenced quantum information processing, while simultaneously presenting challenges to the philosophical implications of quantum theory~\cite{Grinbaum_2017}. Many DI protocols are structured to incorporate the violation of a Bell inequality, necessitating two spatially separated and non-communicating players to share an entangled quantum state for successful protocol implementation. Ensuring these conditions makes the practical realisation of DI protocols a formidable challenge.

While a necessary ingredient for DI quantum protocols is entanglement, the BB84 QKD protocol does not use this resource. The core of the BB84 protocol lies in encoding and exchanging qubits in different bases to enable secure key distribution through quantum channels while detecting potential eavesdropping attempts. Hence, while the players do not need to share any entanglement, they need a quantum channel, and should be able to prepare certain states and to measure them. Motivated by the BB84 protocol, the prepare-and-measure scenario was developed as an alternative allowing for entanglement-free protocols~\cite{WP15,rusca2019quantum,brask2017megahertz}. The prepare-and-measure scenario does not allow for DI conclusions, as it exposes non-classicality of quantum system only if certain assumptions are made about it. The most common assumption is to bound the dimension of the underlying system. The scenario in which the dimension of the underlying Hilbert space is bounded, but preparations and measurements are not trusted, is known as a semi-device-independent scenario (SDI) \cite{Gallego_2010, Paw_owski_2011}.

Despite being relatively simple to implement, prepare-and-measure protocols have received significantly less attention than their entanglement-based counterparts, particularly outside the context of QKD. While numerous tools have been developed to scrutinize different correlation sets in the Bell scenario, along with the significance of extremal correlations and their self-testing properties, a comparable emphasis on correlation sets in prepare-and-measure scenarios is lacking. The primary objective of this work is to highlight the connection between the power of quantum correlations in the device-independent (DI) non-local and the semi-device-independent (SDI) prepare-and-measure scenarios.

The exploration of the equivalence between these two scenarios, particularly in the context of QKD protocols, has been undertaken in prior works such as~\cite{pawlowski2011semi} and~\cite{WP15}. While a plethora of tools and methods have been developed to explore the DI scenario, yielding a wealth of theoretical results, understanding when and how these results can be mapped onto the prepare-and-measure scenario remains a compelling question. In this context, we present a  procedure for translating a Bell inequality into a prepare-and-measure scenario. We discuss the properties that can be translated using such a protocol and the conditions under which such translations hold. This endeavour aims to contribute to a more comprehensive understanding of the relationship between quantum correlations in different bipartite scenarios.

On another level, the temporal structure serves as a distinguishing factor between non-locality-based and prepare-and-measure scenarios. Bell tests can be formulated to enable simultaneous operations by all players, representing a parallel execution. In contrast, prepare-and-measure scenarios exemplify time-ordered or sequential protocols. It is noteworthy that non-classicality can manifest in both parallel and sequential scenarios. For instance, quantum computational advantage has been associated with Kochen-Specker contextuality as a parallel phenomenon~\cite{Howard_2014, Delfosse, Juani}, and also linked to the transformation contextuality as a sequential phenomenon~\cite{mansfield2018quantum}. Despite these insights, a comprehensive understanding of the relationships between these diverse results is yet to be established. Our perspective suggests that a fruitful avenue towards this enhanced comprehension involves identifying rigorous connections between parallel and sequential quantum phenomena. Our work contributes to the exploration of this problem by focusing on non-locality as a representative of parallel quantum phenomena and prepare-and-measure non-classicality as an example of sequential phenomena. 

The paper is organized as follows. In Sec.~\ref{sec:Bell}, we review the key concepts related to Bell scenarios. In Sec.~\ref{sec:PM} we introduce the prepare-and-measure scenario we consider and define theoretical objects, which are used in Sec.~\ref{sec:translating} to establish a certain form of mapping between the two scenarios. Sec.~\ref{sec:examples} presents several examples illustrating how well-known Bell inequalities can be reformulated within the prepare-and-measure framework. These examples demonstrate the versatility and applicability of our translation approach. In Sec.~\ref{sec:app}, we apply our translation method to achieve several results: we demonstrate self-testing of any qubit measurement in the prepare-and-measure scenario, define a symmetric prepare-and-measure scenario by translating the bilocality inequality, and explore how entanglement can be exchanged for communication in various network-device-independent certification protocols. Finally, in Sec.~\ref{sec:concl}, we provide concluding remarks.

\section{Bell scenario}\label{sec:Bell}

In the bipartite Bell scenario two spatially separated observers, denoted as Alice and Bob, share a bipartite quantum state $\ket{\psi}$. Each observer receives a classical input, $x\in\mathcal{X}$ for Alice and $y\in\mathcal{Y}$ for Bob, specifying the measurement choice, and yields a classical output, $a\in\mathcal{A}$ for Alice and $b\in\mathcal{B}$ for Bob, representing the measurement result. The cardinality of all sets is arbitrary but finite. The quantum state and the associated measurements are neither characterized, nor trusted. The measurement output correlations derived from this setup can be expressed through the correlations in joint conditional probabilities
\begin{equation}
    p(a,b|x,y) = \bra{\psi}M_x^a\otimes N_y^b\ket{\psi},
\end{equation}
where $M_x^a$ is Alice's measurement operator corresponding to input $x$ and output $a$, while $N_y^b$ is Bob's measurement corresponding to input $y$ and output $b$. These measurements are modelled in the most general way as positive-operator-valued measures (POVMs): $M_x^a,N_y^b \geq 0$, $\sum_aM_x^a = \idd$, $\sum_bN_y^b = \idd$. The set of conditional probability vectors $\vec{p} = \{p(a,b|x,y)\}_{a,b,x,y}$, denoted by $\mathcal{Q}$, is convex.
Noticeably, this set contains correlations that can be replicated through purely classical means. These correlations are characterised as accommodating a local hidden variable (LHV) model and are commonly referred to as local or classical correlations.
However, quantum theory also presents correlations that defy explanation through LHV models~\cite{Bell64}.  The boundary between quantum and classical correlations is characterized by Bell inequalities.

\begin{defi}\label{BF}
    A Bell functional is any linear combination of the observed probabilities,
    \begin{equation}\label{bi}
    I^{\text{NL}} = \sum_{abxy}k_{xy}^{ab}p(a,b|x,y),
\end{equation} where $k_{xy}^{ab}$ are real coefficients.
\end{defi}
For correlations that can be reproduced by LHV models, the Bell functional cannot be higher than a certain value, known as the local bound. In other words, LHV models satisfy Bell inequalities.

\begin{defi}\label{BellBound}
    The quantum bound of a Bell inequality is the maximal value of the Bell functional over all quantum correlations in $\mathcal{Q}$
    \begin{equation}\label{qb}
    b_q^{\text{NL}} = \max_{\vec{p}\in\mathcal{Q}}I^{\text{NL}}.
\end{equation}
\end{defi}
To analyze these quantum bounds, it is useful to introduce the associated Bell operator.
\begin{defi}\label{BO}
    Given a Bell functional as above, the corresponding Bell operator is defined as
    \begin{equation}\label{eq:BellOperator}
    \mathcal{B}^{\text{NL}} = \sum_{abxy}k_{xy}^{ab}M_{x}^a\otimes N_y^b.
\end{equation}
\end{defi}
The quantum bound $b_q^{\text{NL}}$ is equal to the maximum expectation value of this operator over all normalized bipartite quantum states.

Alice's $d$-output measurements $\{M_{x}^a\}_{a=0}^{d-1}$, usually described as a POVM,  can be equivalently specified by a set of unitary measurement observables $\{A_x^{(0)}=\mathds{1}, A_x^{(1)}, \dots, A_x^{(d-1)}\}$, called observables and defined as 
\begin{equation}\label{eq:observable_d}
    A_x^{(k)} = \sum_{a=0}^{d-1} \omega^{ak} M_{x}^a =  \sum_{a=0}^{d-1} e^{2i \pi \frac{ak}{d}} M_{x}^a, 
\end{equation}
for $k=0,\cdots,d-1$.
Analogously we can define Bob's measurement $\{N_{y}^b\}_{b=0}^{d-1}$ through a similar $d$-tuple:
\begin{equation}
    \label{eq:observable_dBob}
    B_y^{(l)} = \sum_{b=0}^{d-1} \omega^{bl} N_{y}^b =  \sum_{b=0}^{d-1} e^{2i \pi \frac{bl}{d}} N_{y}^b.
\end{equation}
Alice's unitary observables satisfy the following two properties
\begin{subequations}
  \begin{gather}\label{cond1}\left[A_x^{(k)}\right]^\dagger = A_x^{(d-k)}, \\ \label{cond2}\left[A_x^{(k)}\right]^\dagger A_x^{(k)} = \idd,\end{gather}\end{subequations}
for all $k$ and $x$, and equivalently for Bob. 
    
Instead of using the probability vector $\{p(a,b|x,y)\}$, the correlations in the Bell scenario can also be characterised through the so-called correlators.

\begin{defi}
 The quantum correlator is defined as
 \begin{align}\nonumber
      \mathrm{corr^{\text{NL}}}(x,y,k,l) &\equiv \bra{\psi} A_x^{(k)} \otimes  B_y^{(l)} \ket{\psi} \\ \nonumber       &=\sum_{ab} \omega^{ak+bl} \bra{\psi} M_{x}^a \otimes  N_{y}^b \ket{\psi} \\ \label{correlatorprob} &=\sum_{ab} \omega^{ak+bl} p(ab|xy).
\end{align}
\end{defi}

A general Bell operator~\eqref{eq:BellOperator} can be written as a linear combination of unitary operators corresponding to measurement observables: 
\begin{align}\label{eq:NL-BellD}
     \mathcal{B}^{\mathrm{NL}}(A_x^{(k)},B_y^{(l)})  = \sum_{xykl} c_{xy}^{kl} A_x^{(k)} \otimes B_y^{(l)},
\end{align}
with $c_{xy}^{kl}$ complex coefficients.

\textbf{Self-testing. }In some scenarios, the quantum bound of a Bell inequality is achieved by a single quantum strategy, defined by a specific shared state and set of measurements, unique up to local isometries and complex conjugation. When this is the case, the quantum bound is said to self-test the corresponding quantum strategy~\cite{MayersYao,vsupic2020self}. Self-testing serves as a critical primitive for device-independent certification of quantum states and measurements~\cite{govcanin2022sample,dos2024experimental}. This bound is considered a robust self-test if any quantum strategy achieving a violation that is $\epsilon$-far from it must necessarily be $f(\epsilon)$-close to the strategy it self-tests.   

A powerful algebraic tool in the analysis of quantum bounds and self-testing is the sum-of-squares (SOS) decomposition. 

\begin{defi}\label{SOS}
    Given a Bell operator $\mathcal{B}^{\text{NL}}$ with quantum bound $b_q^{\text{NL}}$, a sum-of-squares (SOS) decomposition of the shifted Bell operator is an identity of the form
    \begin{multline}\label{eq:Bell-sos}
    {b}_q\idd - \mathcal{B}^{\mathrm{NL}}(A_x^{(k)},B_y^{(l)}) =  \\ = \sum_m \left[f_m(A_x^{(k)},B_y^{(l)})\right]^\dagger\left[f_m(A_x^{(k)},B_y^{(l)})\right],
\end{multline}
where $f_m$ are non-commutative polynomials of the measurement observables observables.
\end{defi}
Any quantum strategy $\{\ket{\psi},A_x^{(k)},B_y^{(l)}\}$ reaching the quantum bound $b_q$ satisfies
\begin{equation}\label{eq:sosterms}
f_m(A_x^{(k)},B_y^{(l)})\ket{\psi} = 0
\end{equation}
for every $m$. These expressions are commonly used to put constraints on the strategy achieving the maximal violation and reaching self-testing statements.

SOS decompositions are useful for obtaining robust self-testing statements as well. Consider a violation $b_q-\epsilon$; from the SOS decomposition above we derive the following constraints  
\begin{equation}\label{eq:sosterms_robust}
\| f_m(A_x^{(k)},B_y^{(l)})\ket{\psi} \| \leq \epsilon.
\end{equation}
This set of equations characterises the strategies which lead to a violation $\epsilon$-close to the optimal one. 

\section{Prepare-and-measure scenario}\label{sec:PM}
Let us now consider the standard prepare-and-measure scenario shown in Fig. \ref{fig:NLPM}, which also exposes a classical-quantum separation when communication is allowed under some assumptions.

 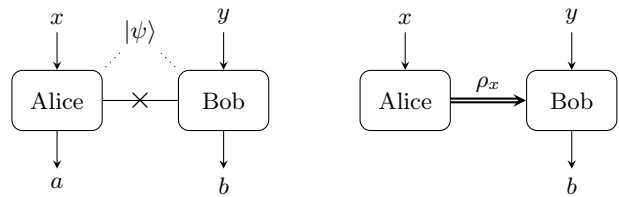
\begin{figure}[htbp]
      \centering
        \begin{center}
\begin{tikzpicture}[
    every node/.style={align=center},    >=stealth
]
    \node[draw, rectangle, rounded corners, minimum width=1.2cm, minimum height=0.8cm] (Alice_NL) {Alice};
    \node[draw, rectangle, rounded corners, minimum width=1.2cm, minimum height=0.8cm, right=1cm of Alice_NL] (Bob_NL) {Bob};
    
    \node[draw, rectangle, rounded corners, minimum width=1.2cm, minimum height=0.8cm, right=1.2cm of Bob_NL] (Alice_PM) {Alice};
    \node[draw, rectangle, rounded corners, minimum width=1.2cm, minimum height=0.8cm, right=1cm of Alice_PM] (Bob_PM) {Bob};

    \draw[->] (Alice_NL.north) ++(0,0.5cm) node[above] {$x$} -- (Alice_NL.north);
    \draw[->] (Bob_NL.north) ++(0,0.5cm) node[above] {$y$} -- (Bob_NL.north);
    
    \draw[->] (Alice_PM.north) ++(0,0.5cm) node[above] {$x$} -- (Alice_PM.north);
    \draw[->] (Bob_PM.north) ++(0,0.5cm) node[above] {$y$} -- (Bob_PM.north);

    \draw[->] (Alice_NL.south) -- ++(0,-0.5cm) node[below] {$a$};
    \draw[->] (Bob_NL.south) -- ++(0,-0.5cm) node[below] {$b$};
    
    \draw[->] (Bob_PM.south) -- ++(0,-0.5cm) node[below] {$b$};

    \draw (Alice_NL.east) -- (Bob_NL.west) node[midway] (X1) {\large $\times$};
    \draw[double,thick, ->] (Alice_PM.east) -- (Bob_PM.west) node[midway, above] {$\rho_x$};

    \node at ($(Alice_NL.north east)!0.5!(Bob_NL.north west)+(0,0.5cm)$) (*) {$\ket{\psi}$};
    \draw[dotted] (*) -- (Alice_NL.north east);
    \draw[dotted] (*) -- (Bob_NL.north west);

\end{tikzpicture}
\end{center}
        \caption{Two standard bipartite correlations structures: the Bell on the left and the prepare-and-measure scenario on the right side. 
        }
        \label{fig:NLPM}
    \end{figure}
    
In this case, Alice receives a classical input $x$, prepares a message and sends it to Bob. This message can be quantum, and is thus described with a density matrix $\rho_x$. After the communication stage has ended, Bob receives an input $y$ and produces an output $b$, which is usually interpreted as a guess of Alice's original input. Bob obtains this output by applying a measurement $N_y^b$. Thus, the scenario is characterised by the following probability distribution
\begin{equation}
    P(b|x,y) = \Tr\left[N_{y}^b\rho_x\right].
\end{equation}
In the remainder of this paper probabilities written with an upper-case letter $P$ will indicate that they are obtained in the prepare-and-measure scenario, while probabilities with a lowercase letter $p$ are those native to the Bell scenario.

If there is no restriction on the communication channel used by Alice and Bob, they can reproduce any desired probability distribution using only classical resources.
A non-trivial separation between classical and quantum correlations can be made only if the communication between Alice and Bob is restricted in some way.
A standard constraint is to upper bound the dimension of the communication channel, the most relevant example being the study of quantum random access codes~\cite{Wiesner,Ambainis}. In this case, it is valid to ask what are the limits of quantum strategies, and whether certain quantum strategies can be self-tested. There has been a significant amount of literature devoted exactly to this question~\cite{TKV+18,FK19,TSV+20,Tav20,MO21,DPV24,DPTV24,NPVA23}. 

However, bounding the dimension of the communication channel is not the only way to introduce communication restrictions in the prepare-and-measure scenario; for other possibilities see~\cite{VanHimbeeck2017semidevice, PPT24}.
Similarly to the nature of communication restrictions, there exist more degrees of flexibility in defining the prepare-and-measure scenario. For instance, Alice's inputs can be divided into two distinct sets, with inputs in each set sampled according to two mutually independent probability distributions, and designated as $x$ and $a$ respectively. In this setup, where Alice prepares states denoted as $\rho_{a,x}$, the examination of preparation contextuality becomes relevant~\cite{Spekkens}. In preparation contextuality, communication remains unbounded, yet the states Alice prepares must conform to the preparation equivalence condition: $\sum_ap(a|x)\rho_{a,x} = \sum_ap(a|x')\rho_{a,x'}$, for every pair of distinct inputs $x$ and $x'$. An alternative variation of the prepare-and-measure scenario grants Alice additional autonomy: upon receiving $x$, she independently selects $a$, without a need to adhere to a specific probability distribution, outputs it, and subsequently transmits the state $\rho_{a,x}$.

Numerous efforts have been made to establish an equivalence between non-locality based and prepare-and-measure protocols. 
For instance, in \cite{WF23} the authors present a one-to-one equivalence, when the states that Alice prepares satisfy specific preparation equivalences.
Furthermore, Catani \emph{et al.} demonstrated in ~\cite{CFE+22} that any non-local XOR game (corresponding to correlation Bell inequalities) can be transformed into a prepare-and-measure game with the same classical and quantum bounds, under the constraint that Alice and Bob exclusively utilise reversible operations.

Our approach builds on this foundation: for a fixed Bell framework, we design a corresponding prepare-and-measure scenario such that the quantum bound of the associated functional and its self-testing properties are preserved.
As we discussed earlier, there is some freedom in defining the prepare-and-measure scenario, which is reflected in different assumptions and properties of the mapping.
Unlike \cite{WF23} and \cite{CFE+22}, our mapping is not bijective but it is one-way, starting with a Bell test and constructing a prepare-and-measure game with desirable properties as the outcome.
This is the interesting direction for us, since our original motivation is to systematically translate self-testing protocols from the well-studied non-local framework to prepare-and-measure experiments.

\begin{figure}[h!]
    \centering
    \begin{center}
\begin{tikzpicture}[
    every node/.style={align=center},    >=stealth
]
    \node[draw, rectangle, rounded corners, minimum width=1.2cm, minimum height=0.8cm] (Alice_NL) {Alice};
    \node[draw, rectangle, rounded corners, minimum width=1.2cm, minimum height=0.8cm, right=1cm of Alice_NL] (Bob_NL) {Bob};
    
    \node[draw, rectangle, rounded corners, minimum width=1.2cm, minimum height=0.8cm, right=1.2cm of Bob_NL] (Alice_PM) {Alice};
    \node[draw, rectangle, rounded corners, minimum width=1.2cm, minimum height=0.8cm, right=1cm of Alice_PM] (Bob_PM) {Bob};

    \draw[->] (Alice_NL.north) ++(0,0.5cm) node[above] {$x$} -- (Alice_NL.north);
    \draw[->] (Bob_NL.north) ++(0,0.5cm) node[above] {$y$} -- (Bob_NL.north);
    
    \draw[->] (Alice_PM.125) ++(0,0.5cm) node[above] {$x$} -- (Alice_PM.125);
    \draw[->] (Alice_PM.50) ++(0,0.5cm) node[above] {$a$} -- (Alice_PM.50);
    \draw[->] (Bob_PM.north) ++(0,0.5cm) node[above] {$y$} -- (Bob_PM.north);

    \draw[->] (Alice_NL.south) -- ++(0,-0.5cm) node[below] {$a$};
    \draw[->] (Bob_NL.south) -- ++(0,-0.5cm) node[below] {$b$};
    
    \draw[->] (Bob_PM.south) -- ++(0,-0.5cm) node[below] {$b$};

    \draw (Alice_NL.east) -- (Bob_NL.west) node[midway] (X1) {\large $\times$};
    \draw[double,thick, ->] (Alice_PM.east) -- (Bob_PM.west) node[midway, above] {dim $d$};

    \draw[ultra thick, ->] (3.2,0) -- (3.7,0);

\end{tikzpicture}
\end{center}
    \caption{On the left, the standard non-local scenario, maximised by a maximally entangled state of dimension $d$. On the right, its translated prepare-and-measure version, where Alice is now receiving two inputs and communicates with Bob through a dimensionally bounded quantum channel.}
    \label{fig:ourmapping}
\end{figure}
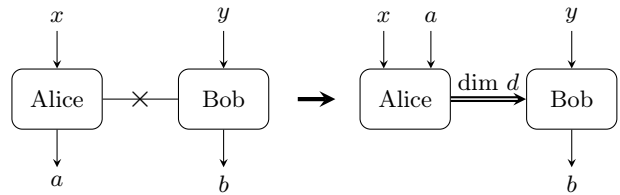

\textbf{The prepare-and-measure scenario in focus. } We now precisely define the prepare-and-measure (PM) scenario we investigate and outline explicitly the mapping procedure that translates certain Bell inequalities into corresponding PM protocols.

Consider a bipartite Bell inequality maximally violated by a maximally entangled state of local dimension $d$ and projective measurements. The scenario originally involves four classical labels: two inputs $x, y$, and two outputs $a, b$. Our PM setup retains exactly these labels and their cardinalities, but adjusts their interpretation: labels $x$ and $a$ now represent input choices for Alice, while Bob maintains his input $y$ and output $b$ (see Fig~\ref{fig:ourmapping} and Table~\ref{tab:notation}). We impose the constraint that Alice’s preparations and Bob’s measurements occur within Hilbert spaces of exactly dimension $d$, corresponding to the minimal dimension necessary to achieve the maximal quantum violation in the original Bell scenario.

To ensure a direct and consistent comparison between the Bell and PM scenarios, we adopt a fixed marginal probability distribution ${P(a|x)}$.  This marginal is chosen to match exactly the distribution $p(a|x)$ arising from the optimal quantum strategy achieving the maximal Bell violation. Since we focus on inequalities maximally violated by maximally entangled states and sharp (projective) measurements, the optimal marginal distributions are uniform. Hence, we fix $P(a|x) = 1/|\mathcal{A}|=1/d$. This standardization ensures a consistent framework for the comparative analysis of non-local and PM scenarios. 

We restrict our analysis to full-correlation Bell inequalities, which involve only joint correlators and exclude marginal terms. This choice is essential for the mapping to PM scenarios, where Alice prepares quantum states, mimicking the effect of her local measurements on an entangled state, and Bob performs measurements on these states. Marginal terms in Bell inequalities correspond to situations where only one party acts (e.g., Bob measuring without an associated preparation from Alice), which has no meaningful counterpart in the PM framework.

\begin{table}[htbp]
\centering
\begin{tabular}{ccc}
\toprule
 & non-local & prepare-and-measure \\
\toprule
Alice's input & $x\in\mathcal{X}$ & $(x,a)\in\mathcal{X} \cross \mathcal{A}$\\
\midrule
Alice's output & $a\in\mathcal{A}$ & $\rho_{a,x}$  \\
\midrule
Bob's input & $y\in\mathcal{Y}$ & $y\in\mathcal{Y}$,  $\rho_{a,x}$  \\
\midrule
Bob's output & $b\in\mathcal{B}$ & $b\in\mathcal{B}$\\
\midrule
shared resource & $\ket{\psi} \in \mathcal{H}^A \otimes \mathcal{H}^B$ & $\{\rho_{a,x}\} \in \mathbb{C}^d$\\
\midrule
Alice's measurements & $\{M_x^a\}$ & -\\
\midrule
Bob's measurements & $\{N_y^b\}$ & $\{N_y^b\}$\\
\bottomrule
\end{tabular}
\caption{Notation for the Bell and prepare-and-measure scenarios. Using the same labels by changing their meaning allows us to compare systematically the two structures.}
\label{tab:notation}
\end{table}

\textbf{Defining the prepare-and-measure functional. } Using Bayes' rule, we define conditional probabilities analogous to those in the Bell scenario:
\begin{align}\label{PMprob}
    P(a,b|x,y) &= {P(b|a,x,y)}{P(a|x)}\\
    &= \frac{1}{|\mathcal{A}|}\tr\left[N_y^b\rho_{a,x}\right],
\end{align}
where $\rho_{a,x}$ denotes the state prepared by Alice given inputs $a$ and $x$, and $N_y^b$ is Bob’s POVM element corresponding to input $y$ and outcome $b$. We define $\mathcal{Q}_{seq}$ to be the set of all quantum correlations that can be decomposed in this way.
In direct analogy to the nonlocal (Bell) scenario, we define PM correlators as follows.
\begin{defi}\label{PMcorr}
    The prepare-and-measure correlators are defined as
    \begin{equation}
    \mathrm{corr^{\text{PM}}}(x,y,k,l) 
= \sum_{ab} \omega^{ak + bl} P(ab|xy),
\end{equation}
where $\omega = \exp{\frac{2i\pi}{d}}$.
\end{defi}

These correlators allow for a direct parallel definition of a PM analogue to the Bell functional.

\begin{defi}\label{pmBF}
    Given a Bell functional $I^{NL}$ defined by coefficients $k_{xy}^{ab}$, the corresponding prepare-and-measure functional is defined as
    \begin{equation}\label{eq:PMfunctional}
    I^{\text{PM}} = \sum_{abxy}k_{xy}^{ab}P(a,b|x,y),
\end{equation}
where $P(a,b|x,y)$ are prepare-and-measure correlation probabilities defined in Eq.~\eqref{PMprob}. Equivalently, expressed via PM correlators (Def.~\ref{PMcorr}) it has the form 
\begin{equation}
     I^{\text{PM}} = \sum_{xykl}c_{xy}^{kl}\mathrm{corr^{\text{PM}}}(x,y,k,l),
\end{equation}
with appropriate coefficients $c_{xy}^{kl}$.
\end{defi}

Analogously to the quantum bound of Bell functionals we can define the quantum bound of associated PM functional:

\begin{defi}\label{boundPM}
   The quantum bound of the PM functional is the maximal value of the PM functional over all quantum correlations $\mathcal{Q}_{seq}$
\begin{align*}
    b_q^{\text{PM}}= \max_{\Vec{P} \in \mathcal{Q}_{seq}} \mathcal{I}^{\text{PM}}.
\end{align*} 
\end{defi}

This quantum bound will be at least as big as the quantum bound of the original Bell inequality. To see that, it suffices to consider Alice preparing the following set of states:
\begin{equation}\label{eq:steered-states}
    \rho_{a,x} = \frac{\text{tr}\left[ (M_x^a \otimes \mathds{1}) \ket{\psi}\bra{\psi}\right]}{P(a|x)}.
\end{equation}
These states trivially satisfy the condition
\begin{equation}\label{eq:no-signalling-set}
    \sum_a P(a|x) \rho_{a,x} = \sigma \qquad \forall x
\end{equation}
which we will refer to as the non-signalling condition; it also corresponds to the preparation equivalence constraint in \cite{WF23}. Sets of states that violate the non-signalling condition lack a Bell ancestor, raising the prospect that the quantum bound of the functional in the prepare-and-measure context might exceed the quantum bound of the corresponding Bell inequality. 

\textbf{Quasi-observables. } To systematically characterize the states Alice prepares, we introduce the concept of quasi-observables.
\begin{defi}\label{def:QO}
    In the context of prepare-and-measure scenarios, quasi-observables are defined by the discrete Fourier transform of Alice’s preparations:
    \begin{equation}
  \tilde{\rho}_x^{(k)} = \sum_{a=0}^{d-1}\omega^{ak}{\rho}_{a,x}.
\end{equation} 
\end{defi}
These quasi-observables relate directly to the PM correlators through the identity
\begin{align}\notag
\mathrm{corr^{\text{PM}}}(x,y,k,l) 
&= \sum_{ab} \omega^{ak + bl} P(ab|xy) \\ \notag
&= \sum_{ab} \omega^{ak + bl} P(a|x)P(b|axy)\\ \notag
&= \sum_{ab} \omega^{ak + bl} P(a|x)\Tr{N_{y}^b \rho_{a,x}}\\
&=   \frac{1}{d}\Tr{B_y^{(l)}\tilde{\rho}_x^{(k)}}.\label{eq:PM-correlator}
\end{align}
This explains their naming convention, highlighting their structural analogy to observables in Bell scenarios. Quasi-observables satisfy a hermiticity condition analogous to Eq.~\eqref{cond1}:
\begin{equation*}
    \left[\tilde{\rho}_x^{(k)}\right]^\dagger = \tilde{\rho}_x^{(d-k)}.
\end{equation*}
However, they do not in general satisfy a full analogue of the condition for unitary observables. Instead, they obey the weaker constraint summarized in the following lemma:
\begin{lem}\label{lemdits}
    For every input $x$, Alice's quasi-observables satisfy:
    \begin{align}\label{eq:weaker-condition-first}
    \text{tr} \left( \sum_{n=1}^{d-1} \rho^{(n)} \rho^{(-n)}\right) \leq d (d-1). 
\end{align}
\end{lem}
The proof is provided in App.~\ref{app:qudits-weakercondition}.

For qubits, quasi-observables simplify significantly. Recall that in the Bloch representation a density matrix can be written as $\rho = (\idd + \Vec{m}\cdot\Vec{\sigma})/2$, where $\Vec{\sigma} = (\sigma_z,\sigma_x,\sigma_y)$ is a vector of Pauli matrices and $\Vec{m}$ is a real vector with norm at most equal to $1$. Using this representation for quasi-observables, by taking $\rho_{a,x} = (\idd + \vec{m}_{a,x}\cdot\vec{\sigma})/2$ we get
\begin{align}\label{eq:fake-observable-2}
    \tilde{\rho}_x =   \frac{\Vec{m}_{0,x}-\Vec{m}_{1,x}}{2}\cdot \Vec{\sigma} = \Vec{v}_x \cdot \Vec{\sigma},
\end{align}
Hence, qubit quasi-observables are Hermitian, and satisfy
\begin{equation}
 \tilde{\rho}_x^\dagger\tilde{\rho}_x = \|\vec{v}_x\|^2\idd  \leq \idd  
\end{equation}
 This characterization is crucial for identifying classes of Bell inequalities whose quantum bounds remain unchanged upon translation to the PM scenario.

Finally, we introduce the PM equivalent of the Bell operator.

\begin{defi}\label{pmBO}
    The PM Bell-like operator corresponding to the PM functional (Def.~\ref{pmBF}) is
    \begin{align}\label{eq:PM-Bell}
    \mathcal{B}^{\text{PM}}\left(\tilde{\rho}_x^{(k)},B_y^{(l)}\right) =\frac{1}{d} \sum_{x,y,k,l} c_{xy}^{kl} \tilde{\rho}_{x}^{(k)} B_y^{(l)}.
\end{align}
\end{defi}

In terms of the PM Bell-like operator the quantum bound of the PM functional is
\begin{align*}
    b_q^{\text{PM}} = \max_{\Tilde{\rho},B} \text{tr} \left[ \mathcal{B}^{\text{PM}}\left(\tilde{\rho}_x^{(k)},B_y^{(l)}\right) \right].
\end{align*} 

The task of comparing the quantum bound of a Bell inequality with that of its prepare-and-measure counterpart is fundamentally translated into assessing the operator norm of the Bell operator~\eqref{eq:NL-BellD} against the trace of its prepare-and-measure analogue~\eqref{eq:PM-Bell}. This entails a comparison between the dominant eigenvalue associated with a tensor product structure and the trace of a dot product of two operators.

\section{Quantum bound of prepare-and-measure functionals}\label{sec:translating}  We now discuss how to translate Bell inequalities maximally violated by maximally entangled states and projective measurements into the prepare-and-measure (PM) scenario. Showing that a translated functional preserves the quantum bound can be significantly facilitated whenever the quantum bound of the original Bell functional can be proven by using the SOS decomposition. Following its definition (Def.~\ref{SOS}), we describe how the same algebraic structure can be transferred to the PM scenario by rearranging the terms in the PM functional to mimic the Bell scenario explicitly.

To relate the quantum bounds in the Bell scenario (Def.~\ref{BellBound}) and the PM scenario (Def.~\ref{boundPM}), we employ the following identity known as the "swap trick", which holds for maximally entangled state $\ket{\phi_+}$
\begin{align}\label{eq:swap-trick}
    \bra{\phi^+} A \otimes B^T \ket{\phi^+}= \frac{1}{d} \text{tr}(A \cdot B)
\end{align}
where $d$ is the Hilbert space dimension. This identity has previously been employed to validate the security of certain prepare-and-measure Quantum Key Distribution (QKD) protocols~\cite{WP15}. 

Using the swap trick, we directly express the trace of the PM Bell operator as the expectation value of a corresponding Bell operator evaluated on the maximally entangled state:

\begin{align}\label{eq:tr-and-maxent}
\tr\left[\mathcal{B}^{\text{PM}}\left(\tilde{\rho}_x^{(k)},B_y^{(l)}\right)\right] = 
     \bra{\phi^+} \mathcal{B}^{\text{NL}}\left(\tilde{\rho}_x^{(k)},B_y^{(l)}\right) \ket{\phi^+}.
\end{align}
Note that the definition of the PM operator (Eq.~\eqref{eq:PM-Bell}) explicitly includes the factor $1/d$, matching exactly the factor arising from the swap trick. Eq.~\eqref{eq:tr-and-maxent} reduces the problem of determining the quantum bound of a PM functional to the simpler task of evaluating the expectation value of a Bell-type operator on a maximally entangled state. However, this PM-related Bell-type operator differs from the original Bell operator, as it uses the quasi-observables $\tilde{\rho}_x^{(k)}$ rather than the unitary observables $A_x^{(k)}$. Thus, calculating the quantum bound requires accounting for differences between these two sets of operators, as discussed following Definition~\ref{def:QO}.

For qubit quasi-observables, due to their particularly simple structure, we can rigorously establish that the quantum bound remains unchanged, as formalized by the following theorem:

\begin{thm}\label{thm}
     Let $I^{NL}$ be a full-correlation Bell functional whose quantum bound $b_q^{NL}$
  is achieved by projective measurements performed on a maximally entangled pair of qubits. Then, the associated prepare-and-measure functional $I^{PM}$ has exactly the same quantum bound, that is, $b_q^{PM} = b_q^{NL}$.
\end{thm}

\begin{proof}
For a full-correlation Bell inequality maximally violated by qubits, it is known that the inequality can be cast as an XOR game. For XOR games, the SOS decomposition has a particular structure~\cite{cui2024computational,baroni2024quantum}: the polynomials \(f_m\) split naturally into two categories. The first set of polynomials, indexed by Alice’s measurements \(x\), take the form \(f_x(A_x,B_y) = A_x - f'_x(B_y)\), where \(f'_x(B_y)\) depend solely on Bob’s observables. The second set involves polynomials only in Bob’s observables. When observables correspond to projective measurements, we have \(A_x^2 = \mathds{1}\) and \(B_y^2 = \mathds{1}\). If the measurements were not projective, the SOS decomposition would have additional positive-semidefinite terms proportional to \(\mathds{1} - A_x^2\) and \(\mathds{1}-B_y^2\). Thus, the original SOS decomposition remains valid, as an inequality, when replacing Alice's projective observables \(A_x\) by quasi-observables \(\tilde{\rho}_x\). Hence, we obtain the inequality:

\begin{equation}\label{dif1}
    b_q^{NL} - \bra{\phi^+} \mathcal{B}^{\text{NL}}\left(\tilde{\rho}_x,B_y\right) \ket{\phi^+} \geq \sum_m\bra{\phi^+}\left[f_m(\tilde{\rho}_x,B_y)\right]^2\ket{\phi^+}.
\end{equation}

Since the right-hand side is non-negative, Eq.~\eqref{eq:tr-and-maxent} implies:
\begin{equation}\label{dif}
    b_q^{NL} -  b_q^{PM} \geq 0.
\end{equation}
On the other hand, by construction, the quantum bound of the PM functional is always at least as large as the quantum bound of the corresponding Bell functional, thus:
\[
    b_q^{PM} \geq b_q^{NL}.
\]

Combining these two inequalities yields the desired equality:
\[
    b_q^{PM} = b_q^{NL},
\]
completing the proof.
\end{proof}

The translation of SOS decompositions from Bell to PM scenarios allows for a direct transfer of self-testing results. Saturation of inequality~\eqref{dif} implies that each squared term in the SOS decomposition vanishes, yielding PM analogues of the conditions given by Eq.~\eqref{eq:sosterms}:
\begin{equation*}
    f_m(\rho_x,B_y)\ket{\phi_+} = 0.
\end{equation*}
These conditions provide a foundation for self-testing statements in the PM scenario. In the next section, we will give concrete examples of such relations inherited from SOS decompositions, and show how they facilitate self-testing in the PM scenario. Furthermore, in all of our examples robustness is also preserved.

In higher dimensions (qudits), no similarly simple characterization of SOS decompositions currently exists. While quasi-observables satisfy the weaker condition established in Lemma~\ref{lemdits}, measurement observables satisfy significantly stronger unitary conditions. Nonetheless, in the next section, we demonstrate through explicit examples of important classes of full-correlation Bell inequalities that the translated SOS decomposition, together with Lemma~\ref{lemdits}, provides enough structure to guarantee preservation of the quantum bound. Given these considerations, we formulate the following conjecture as a natural generalization of Theorem~\ref{thm}:

\begin{cnj}
Consider a full-correlation Bell functional $I^{NL}$ whose quantum bound $b_q^{NL}$ is attained by projective measurements on a maximally entangled pair of qudits. Then, the associated prepare-and-measure functional $I^{PM}$ possesses the same quantum bound $b_q^{PM} = b_q^{NL}$.
\end{cnj}

\section{Examples}\label{sec:examples}
In this section we consider bipartite Bell inequalities maximised by the maximally entangled state and their translated prepare-and-measure version. Case by case, we see that the SOS decomposition can also be translated, meaning that the bound and the optimal strategy are preserved.

\subsection{Qubit strategies}

In the qubit scenario, for each input $x$, there exists a unique quasi-observable $\tilde{\rho}_x^{(1)}$, which we shall denote simply by $\tilde{\rho}_x$. As previously outlined, these operators are Hermitian and satisfy the operator inequality: $\tilde{\rho}_x^\dagger\tilde{\rho}_x \leq \idd$.

Before proceeding to explicit examples of self-testing within the prepare-and-measure (PM) framework, we clarify how specific properties of the quasi-observables enable self-testing of the prepared states. Notably, local unitary transformations performed independently by Alice and Bob cannot be detected solely from the observed PM correlations. Consequently, our characterization of Alice's states is inherently limited to equivalence classes under global rotations.

The purity of the states prepared by Alice can be guaranteed by demonstrating that the quasi-observables satisfy the stronger condition $\tilde{\rho}_x^2 = \idd$. If this identity can be established through an SOS decomposition, the purity of the states is confirmed, giving rise to a clear geometric interpretation of the operators $\tilde{\rho}_x$.

Specifically, the operators $\tilde{\rho}_x$ define the geometry of the $|\mathcal{A}| \cdot |\mathcal{X}| = 2n$ states prepared by Alice (modulo global rotations), which can be parameterized by $n(n+1)/2$ angles:

\begin{itemize}
    \item $n$ angles called $\omega_x$, used to describe the non-orthogonality between each pair labelled by $x$
    \begin{equation}
        \norm{\Vec{v}_x} = \cos(\omega_x),
    \end{equation}
    \item $\binom{n}{2}$ angles $\phi_{xx'}$, describing the angles between the different bases
    \begin{equation}
        \frac{\Vec{v}_x\cdot \Vec{v}_{x'}}{\norm{\Vec{v}_x} \norm{\Vec{v}_{x'}}} = \cos(\phi_{xx'}).
    \end{equation}
\end{itemize}

Determining quasi-observables $\tilde{\rho}_x$ satisfying the equality $\tilde{\rho}_x^\dagger \tilde{\rho}_x = \mathbb{I}$ is thus equivalent to specifying these $n(n+1)/2$ angles, which, in turn, completely characterizes the $2n$ states up to a global unitary rotation.

\textbf{A family of Bell inequalities with binary inputs and outputs}--- 
Consider the simplest bipartite scenario, in which $\left|\mathcal{X}\right| = \left|\mathcal{Y}\right| = \left|\mathcal{A}\right| = \left|\mathcal{B}\right|=2$, \emph{i.e.} all inputs and outputs are binary. A relevant family of Bell inequalities in this scenario is the one introduced in~\cite{Le_2023}, and further discussed in~\cite{barizien2023custom}. The family of Bell inequalities is parametrized with three parameters $\alpha$, $\beta$ and $\gamma$, and the Bell operator is
\begin{multline}\label{param}
    \mathcal{B}_{\alpha\beta\gamma}^{\text{NL}}(A_x,B_y) = \\
    = \cos(\alpha+\beta)\cos(\alpha + \gamma)A_0\otimes (\cos(\gamma)B_0-\cos(\beta)B_1) 
    \\ + \cos(\beta)\cos(\gamma)A_1\otimes (-\cos(\alpha+\gamma)B_0 +\cos(\alpha+\beta)B_1).
\end{multline}
This family includes the CHSH inequality for $\alpha=\gamma=0$ and $\beta=\pi$. A Bell inequality belonging to this family allows for  quantum violation when $\cos(\alpha+\gamma)\cos(\alpha+\beta)\cos(\beta)\cos(\gamma) < 0$. Whenever this is the case reaching the quantum bound
\begin{equation*}
    b_q = \pm \sin(\alpha)\sin(\gamma-\beta)\sin(\alpha+\beta+\gamma)
\end{equation*}
self-tests the maximally entangled pair of qubits~\cite{wooltorton2023deviceindependent}. This property makes it a good candidate for our translation procedure, from which we obtain the following prepare-and-measure Bell-like operator 
\begin{multline*}
    \mathcal{B}_{\alpha\beta\gamma}^{\text{PM}}(\tilde{\rho}_x,B_y) = \\ 
    = \frac{1}{2}\Big[\cos(\alpha+\beta)\cos(\alpha + \gamma)\tilde{\rho}_0(\cos(\gamma)B_0-\cos(\beta)B_1) \\
    + \cos(\beta)\cos(\gamma)\tilde{\rho}_1(-\cos(\alpha+\gamma)B_0 +\cos(\alpha+\beta)B_1)\Big].
\end{multline*}
Adapting the SOS decomposition of the Bell inequality (shown in App. \ref{app:family_binary}), we can prove that the quantum bound is preserved.
The inequality is saturated when $\|\tilde{\rho}_0\| = \|\tilde{\rho}_1\| = 1$, implying the purity of $\rho_{a,x}$ for all $a$ and $x$, and the strategy is given by the one self-testing the Bell inequality, and using the construction from~\eqref{eq:steered-states}. Given that the maximal violation of the Bell inequality is a self-test, when Alice and Bob are restricted to qubit strategies the only strategy, up to a global rotation, is the one that is self-tested in~\cite{wooltorton2023deviceindependent}. This implies that there exists a unitary $U$ such that
\begin{align}
    U\tilde{\rho}_0U^\dagger &= \sigma_x\\
    U\tilde{\rho}_1U^\dagger &= \cos(\alpha)\sigma_x + \sin(\alpha)\sigma_z\\
    UB_0U^\dagger &= \sin(\beta)\sigma_x + \cos(\beta)\sigma_z\\
    UB_1U^\dagger &= \sin(\gamma)\sigma_x + \cos(\gamma)\sigma_z
\end{align}
Notably, for a fixed $\alpha$ there exist values of $\beta$ and $\gamma$ such that the quantum bound of the parametrized Bell inequality is higher than classical, implying that the PM version can self-test any preparation strategy of Alice in which she prepares orthogonal pairs of states for both inputs $x$. Similar results for self-testing in the prepare-and-measure scenario can be found in the literature \cite{MO21}; this result reproduces them in a simpler and more intuitive way.

\textbf{Elegant Bell inequality}----
The Elegant Bell inequality, first introduced in~\cite{Gis07}, considers two parties with dichotomic observables, and asymmetrical input sets, namely $|\mathcal{X}|= 3$ and $|\mathcal{Y}|= 4$. The corresponding Bell operator has the form
\begin{align*}
    \mathcal{B}_{el}^{\text{NL}}\left(A_x,B_y\right) =& A_1 \otimes (B_1 + B_2 - B_3 - B_4) \\
    +& A_2 \otimes (B_1 - B_2 + B_3 - B_4) \\
    +& A_3 \otimes (B_1 - B_2 - B_3 + B_4).
\end{align*}

The quantum bound of the elegant Bell inequality $b_q^{\text{NL}}= 4 \sqrt{3}$ self-tests a maximally entangled pair of qubits~\cite{Acin_2016}. The elegance of this inequality lies in the structure of the self-tested observables.
Namely, Alice's measurements form a set of mutually unbiased bases resembling an octahedron in the Bloch sphere, while Bob's measurements precisely shape the dual Platonic solid, a cube.

When translating the inequality to the prepare-and-measure scenario, it is essential to handle more than two quasi-observables for Alice.
The optimal arrangement is not confined to a single plane on the Bloch sphere, unlike the scenario described above.
The prepare-and-measure operator takes the form
\begin{align}\nonumber
    \mathcal{B}_{el}^{\text{PM}}\left(\tilde{\rho}_x,B_y\right) =\frac{1}{2}\Big[& \tilde{\rho}_1 (B_1 + B_2 - B_3 - B_4) \\ \label{ElBellPM}
    +& \tilde{\rho}_2(B_1 - B_2 + B_3 - B_4) \\
    +& \tilde{\rho}_3(B_1 - B_2 - B_3 + B_4) \Big] \nonumber.
\end{align}
For the reasons explained above, the trace of this operator is equivalent to the expectation value over the maximally entangled state of $\mathcal{B}_{el}^{\text{NL}}\left(\tilde{\rho}_x,B_y\right)$.
Adapting the SOS decomposition of the shifted Bell operator, we retrieve a self-testing statement for the prepare-and-measure version of the task.
More specifically, reaching the quantum bound self-tests anticommutation relations for Bob's operators
\begin{align*}
    \{B_y,B_{y'}\}  = \frac{1}{3} \idd  \qquad \forall y \neq y' \in \mathcal{Y} ,
\end{align*}
and it imposes the following constraints on Alice's quasi-observables
\begin{align*}
    &\tilde{\rho}_x^2 = \idd &\forall x \in \mathcal{X}, \\
    &\{\tilde{\rho}_x,\tilde{\rho}_{x'}\}\ket{\phi_+} = 0 &\forall x \neq x' \in \mathcal{X} .
\end{align*}
This is equivalent to fixing the set of states, parametrised with the following angles
\begin{equation*}
    \omega_x = \pi \qquad \forall x\in \mathcal{X} , \qquad
    \psi_{xx'} = \frac{\pi}{2} \qquad \forall x \neq x'\in\mathcal{X}.
\end{equation*}
The three quasi-observables have exactly the properties of three mutually anti-commuting qubit observables. Since, the states Alice prepares were characterized up to global rotation, this implies that there exists a unitary matrix $U$ such that
\begin{equation}
    U\tilde{\rho}_1U^\dagger = \sigma_z, \quad U\tilde{\rho}_2U^\dagger = \sigma_x, \quad U\tilde{\rho}_3U^\dagger = \pm\sigma_y.
\end{equation}

More detailed calculations and the robust case can be found in App.~\ref{app:EBI}.

\textbf{Chained Bell inequality}---
The chained Bell inequalities~\cite{Pearle,BC} are a generalization of the CHSH inequality with two parties and dichotomic observables, but a larger number of measurements per party, with the Bell operator being
\begin{align*}
    \mathcal{B}^{\text{NL}}_{ch(n)} = \sum_{i=1}^n (A_i \otimes B_i + A_{i+1}\otimes B_i), \qquad  A_{n+1} \equiv - A_1.
\end{align*}
The quantum bound found in~\cite{wehner} is $b_q^{\text{NL}}=2n \cos(\pi/2n)$, and in \cite{SASA16} it was proven that the optimal quantum strategy is realized by a maximally entangled bipartite qubit pair and $n$ measurements per party, maximally spread in one single plane of the Bloch sphere.
The prepare-and-measure functional is
\begin{align*}
    \mathcal{B}^{\text{PM}}_{ch(n)} = \sum_{i=1}^n (\Tilde{\rho}_i \otimes B_i + \Tilde{\rho}_{i+1}\otimes B_i), \qquad  \Tilde{\rho}_{n+1} \equiv - \Tilde{\rho}_1.
\end{align*}
Also in this case, adapting the SOS decomposition of the shifted Bell operator, taken from \cite{SASA16}, one can self-test Bob's observables and Alice's quasi-observables, hence the set of states she prepares.
More precisely, adapting the first order SOS decomposition we prove that, in the optimal configuration, $\rho_i^2 = B_j^2 = \mathds{1}$
in \cite{SASA16}, which implies that the states Alice prepares for each given $x$ are orthogonal. The SOS decomposition of the second order can be used to show that there is a unitary $U$ such that
\begin{align}
    U\tilde{\rho}_xU^\dagger &= \sin\phi_x\sigma_x + \cos\phi_x\sigma_z,\\
    UB_yU^\dagger &= \sin\varphi_y\sigma_x + \cos\varphi_y\sigma_z,
\end{align}
where $\phi_x = (x-1)\pi/n$ and $\varphi_y = (2y-1)\pi/2n$.
More details can be found in App.~\ref{app:chained}.

\subsection{Qudit strategies: Generalised CHSH inequality}

A representative class of the generalization of the CHSH inequality to the scenario with many outputs was introduced in~\cite{KST+19}, by modifying the inequalities presented in~\cite{BuhrmanMasar}. The corresponding Bell functional has the form
\begin{equation*}
    \mathcal{B}_d^{\text{NL}}= \frac{1}{d^3} \sum_n \lambda_n \sum_{x,y} w^{nxy} A_x^{(n)} \otimes  B_y^{(n)},
\end{equation*}
where $\lambda_n$ are complex, finely tuned so that the quantum bound of this Bell inequality is analytically computable and achieved by the maximally entangled bipartite state and two mutually unbiased basis measurements (MUB) for both parties.

The translation of this Bell inequality to the prepare-and-measure scenario preserves the certification properties of the inequality. To be more precise, in~\cite{KST+19} the authors analytically find the quantum bound for every prime dimension $d\geq3$. Except for the case $d=3$ the maximal violation does not self-test the measurements but certifies only that they are mutually unbiased.
For dimension $d=3$ the authors reach a complete self-testing statement, meaning that we can translate it to a complete self-testing statement for its prepare-and-measure version.
The certified resources are MUB observables on Bob's side and a specific set of states obtained by steering the $d$-dimensional maximally entangled bipartite state for Alice's side. 
We refer to~\cite{KST+19} for a more complete characterisation of their results, and to App.~\ref{app:qudit-genCHSH} for technical details on our prepare-and-measure version of the inequality.

Similar techniques can be applied to the inequality presented in \cite{SAT+17}, a generalisation of the chained Bell inequality for higher dimensions.
More details can be found in App.~\ref{app:qudit-genChained}.

\subsection{Limitations of the translation procedure: the case of the Tilted CHSH Inequality} \label{sec:counterexample}

Let us now consider a variation of CHSH, the tilted CHSH inequality introduced in~\cite{AMP12}, corresponding to the Bell operator :
\begin{equation}
    \mathcal{B}^{\text{NL}}_{\alpha \beta} = \alpha A_0 \otimes \mathds{1} + \beta (A_0 + A_1) \otimes B_0 + (A_0 - A_1) \otimes B_1
\end{equation}
with $\alpha \in [0,2[$ and $\beta \geq 1$.
For $\alpha=0$ and $\beta=1$ the CHSH inequality is recovered.
The parameter $\beta >1$ can be interpreted as a non-uniform sampling probability of the variable $y$.
For example, for $\alpha=0$ and $\beta >1$, the inequality is still maximised by a maximally entangled state and our translation protocol can be applied. 

The parameter $\alpha > 0$ has more drastic consequences.
Indeed, in \cite{BP15} the authors provide an SOS decomposition for a general $\alpha$ and $\beta=1$, and they prove that the quantum bound of $\sqrt{8 + 2 \alpha^2}$ is achieved with a partially entangled state
\begin{equation}
    \ket{\phi_\theta} = \cos(\theta) \ket{00} + \sin(\theta) \ket{11},
\end{equation}
where the angle $\theta$ is parametrised by $\alpha$ in the following way
\begin{equation*}
    \alpha \equiv \alpha(\theta) = \frac{2}{\sqrt{1 + 2 \tan^2(2\theta)}}.
\end{equation*}

Note that in this case matching the marginal probabilities optimal for maximally violating the tilted CHSH inequality, $P(a|x)$ is not uniform anymore:
\begin{align*}
    P(a=0|x=0)= \cos^2(\theta), \qquad P(a=0|x=1)= \frac{1}{2}.
\end{align*}
\begin{figure}[h]
    \centering
\begin{tikzpicture}
    \draw (0,0) circle (1cm);
    
    \draw (4,0) circle (1cm);

    \draw[->, line width=0.2mm] (0,0) -- ({1*cos(47)}, {1*sin(47)});
    \draw[->, line width=0.2mm] (0,0) -- ({-1*cos(47)}, {1*sin(47)});
    \draw[->, line width=0.5mm] (0,0) -- (0,1);
    \draw[->, line width=0.1mm] (0,0) -- (0,-1);

    \draw[->, line width=0.2mm] (4,0) -- (5,0);
    \draw[->, line width=0.2mm] (4,0) -- (3,0);
    \draw[->, line width=0.5mm] (4,0) -- (4,1);
    \draw[->, line width=0.1mm] (4,0) -- (4,-1);
\end{tikzpicture}
\caption{Set of states for the prepare-and-measure tilted CHSH inequality, represented on a disc of the Bloch sphere. On the left, the set corresponding to the states steerred from Alice to Bob, in the maximally winning configuration of the non-local game. On the right, an allowed set of states in the prepare-and-measure game, that scores a higher value.
Different thicknesses of the lines represent different sampling probabilities.}
    \label{fig:tiltedCHSH}
\end{figure}
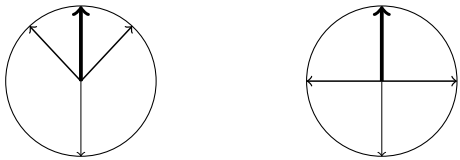

Our proof does not apply to this case, however we can still think of the prepare-and-measure counterpart as in Fig. \ref{fig:ourmapping} by still fixing the dimension to $2$, and enforcing the correct probabilities $p(a|x)$.
It is interesting to see that indeed the set of states that Alice steers to Bob in the Bell scenario will not achieve the highest score in the prepare-and-measure that we constructed.
Consider the simple set of states
\begin{equation*}
\Bigg\{
\begin{array}{l l}
\ket{+}  \text{ for }& (x,a)=(0,1),\\
\ket{-}              & (x,a)=(1,1),\\
\ket{0}              & (x,a)=(0,0),\\
\ket{1}              & (x,a)=(1,0),
\end{array}
\end{equation*}
and observables
\begin{equation*}
    B_0 = \frac{1}{\sqrt{2}} (\sigma_z + \sigma_x) \qquad 
    B_1 = \frac{1}{\sqrt{2}} (\sigma_z - \sigma_x);
\end{equation*}
this configuration reaches a violation of $\alpha \cos(2\theta) + 3 \sqrt{2}$, which is always larger than the quantum bound of the original Bell inequality.
This is corroborating the idea that the condition of maximally entangled states is necessary for this specific mapping to work. In the future it would be interesting to investigate which dimensional bound would be equivalent to partially entangled states.

\subsection{Equivalence between Bell inequalities and Quantum random access codes}

The prepare-and-measure scenario provides a natural framework for investigating quantum random access codes (QRAC)~\cite{QRAC}, a very relevant class of games in the prepare-and-measure scenario, similarly to the role of the CHSH inequality in Bell tests. Within this setting, various self-testing methods have been developed under dimension constraints to identify optimal quantum strategies for successfully implementing these codes~\cite{TKV+18, FK19,farkas2023simple}.

In a general $\smash{n^m \overset{d}{\mapsto} 1}$ quantum random access code (QRAC), $n$ classical messages, each of size $m$ bits, are encoded into a quantum system of dimension $d$. The objective is to retrieve one of these messages, chosen uniformly at random, with a high success probability.
A frequently studied case arises when $m = d$, which we simply denote as $n^m \mapsto 1$ QRAC. It is well-established that encoding information in a quantum system offers advantages over classical systems of the same dimension. This phenomenon is used in various quantum information protocols~\cite{kerenidis2004quantum, Wiesner}.

When applying our translation procedure to some well-known Bell inequalities, we retrieve prepare-and-measure tasks that closely resemble QRACs, but also the Odd Cycle game~\cite{CHTW04}.
In many cases these tasks are, in fact, equivalent up to a relabelling, as summarized in Table \ref{tab:PMversion_game}. Consequently, our procedure plays a role in formally establishing an equivalence between certain well-known Bell inequalities and prepare-and-measure protocols. This unification helps integrating different proof techniques and may offer valuable insights into prepare-and-measure tasks where existing methods fall short.

\begin{table}[ht!]
\centering
\setlength{\tabcolsep}{12pt}
\begin{tabular}{cc}
\toprule
 NL game & PM game \\
  & (up to relabelling) \\
\midrule
CHSH  & QRAC $2^2\to1$ \\
EBI & QRAC $3^2\to1$ \\
Chained  & OddCycle \\
tilted-CHSH with $\alpha=0$ & biased-QRAC $2^2\to1$\\
\bottomrule
\end{tabular}
\caption{Parallelism between Bell inequalities, their prepare-and-measure version and well known prepare-and-measure games.}
\label{tab:PMversion_game}
\end{table}

As an explicit example, we consider the simplest correspondence: the prepare-and-measure CHSH and the QRAC $2^2\to1$.
In the QRAC $2^2\to1$ Alice has to condense the information of two binary inputs, $\{x_0,x_1\}$, in a bi-dimensional resource that she sends to Bob; Bob is asked to correctly retrieve only one of the two bits.
\begin{figure}[h]
    \centering
    \begin{center}
\begin{tikzpicture}[
    every node/.style={align=center},    >=stealth
]
    \node[draw, rectangle, rounded corners, minimum width=1.2cm, minimum height=0.8cm] (Alice_NL) {Alice};
    \node[draw, rectangle, rounded corners, minimum width=1.2cm, minimum height=0.8cm, right=1cm of Alice_NL] (Bob_NL) {Bob};
    
    \node[draw, rectangle, rounded corners, minimum width=1.2cm, minimum height=0.8cm, right=1.2cm of Bob_NL] (Alice_PM) {Alice};
    \node[draw, rectangle, rounded corners, minimum width=1.2cm, minimum height=0.8cm, right=1cm of Alice_PM] (Bob_PM) {Bob};


    \draw[->] (Alice_NL.125) ++(0,0.5cm) node[above] {$x$} -- (Alice_NL.125);
    \draw[->] (Alice_NL.50) ++(0,0.5cm) node[above] {$a$} -- (Alice_NL.50);
    \draw[->] (Bob_NL.north) ++(0,0.5cm) node[above] {$y$} -- (Bob_NL.north);
    
    \draw[->] (Alice_PM.125) ++(0,0.5cm) node[above] {$x_0$} -- (Alice_PM.125);
    \draw[->] (Alice_PM.50) ++(0,0.5cm) node[above] {$x_1$} -- (Alice_PM.50);
    \draw[->] (Bob_PM.north) ++(0,0.5cm) node[above] {$y$} -- (Bob_PM.north);

    \draw[->] (Bob_NL.south) -- ++(0,-0.5cm) node[below] {$b$};
    
    \draw[->] (Bob_PM.south) -- ++(0,-0.5cm) node[below] {$g$};

    \draw[double,thick, ->] (Alice_NL.east) -- (Bob_NL.west) 
    node[midway, above] {dim $2$}
    node[midway, below] {$\rho_{xa}$}
    node[midway, below=1.4cm]{$x \cdot y = a \oplus b$};
    \draw[double,thick, ->] (Alice_PM.east) -- (Bob_PM.west)
    node[midway, above] {dim $2$}
    node[midway, below] {$\rho_{x_0 x_1}$}
    node[midway, below=1.4cm]{$g = x_y$};

\end{tikzpicture}
\end{center}
        \caption{Our prepare-and-measure CHSH on the left and the QRAC $2\to1$ on the right. All labels are binary strings. The structure of the games is identical, but the winning conditions (below the figures) are not equivalent.}
        \label{fig:CHSHQRAC}
\end{figure}
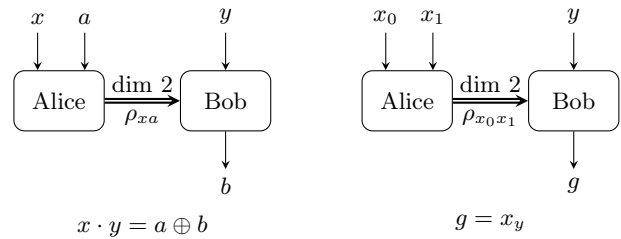
The prepare-and-measure version of CHSH and QRAC $2^2\to1$ share the same structure, but the winning conditions are different, as shown in Fig.~\ref{fig:CHSHQRAC}. Nevertheless we know from our result and the literature on QRACs \cite{TKV+18}, that the classical and quantum bounds are exactly the same; furthermore, the saturation of the quantum bound self-tests the same structures.
This is not a coincidence: an elegant way to explain this is that the two winning conditions, hence the games, become formally the same under a change of variables of the labels of Alice's states:
\begin{center}
\begin{tikzpicture}[scale=0.6]
    \node[] (A) at (0,0.86) {$\rho_{00}$};
    
    \node[] (B) at (3,0) {$\rho_{01}$};
    \node[] (C) at (4,1.7) {$\rho_{10}$};
    \node[] (D) at (5,0) {$\rho_{11}$};

    \path [->] (A) edge [loop right] node {} ();
    \path [->](B) edge node[left] {} (C);
    \path [->](C) edge node[left] {} (D);
    \path [->](D) edge node[left] {} (B);
\end{tikzpicture}    
\end{center}

A similar relabelling can be performed for the other correspondences in Tab.~\ref{tab:PMversion_game}, proving a fundamental equivalence between these games. More details can be found in App.~\ref{app:games}. A more in-depth comparison between Bell inequalities and QRACs can be found in~\cite{TMPB16}.

\section{Applications}\label{sec:app}

We provide now several applications of our method to characterize quantum correlations in the prepare-and-measure scenario.

\subsection{Self-testing any qubit measurement in the prepare-and-measure scenario}\label{sec:allqubitmeasurements}

The prepare-and-measure formulation of the elegant Bell inequality, which is equivalent (up to relabeling) to the $3\to1$ QRAC, can serve as a basis for self-testing any qubit measurement performed by Bob. Specifically, achieving the optimal score in the $3\to1$ QRAC guarantees that, up to a unitary transformation, Alice's preparations correspond to the eigenstates of three Pauli operators. That is, there exists a unitary matrix 
$U$ such that

\begin{equation}\label{tomcom} U\rho_{a,x}U^\dagger = \frac{\mathbb{I} + (-1)^a \sigma_x}{2}. \end{equation}

Self-testing of an arbitrary measurement performed by Bob can be achieved by introducing an additional input to the $3\to1$ QRAC protocol corresponding to the measurement in question. This input, denoted by $y=4$, represents the unknown measurement. The measurement effect associated with output $b$ and input $y=4$ can be expressed as $N_4^b = \xi_b(\idd + \vec{n}_b\vec{\sigma})$, where $\xi_b \geq 0$, and $\sum_b\xi_b = 1$. Since Alice's preparations form a tomographically complete set of states, the correlation probabilities $p(b|x,a,y=4)$ allow the process tomography of the effect $U^\dagger N_4^bU$. These correlations can be expressed as:
\begin{align*}
    p(b|x,a,y=4) &=\tr\left[\rho_{a,x} N_4^b\right] \\&=\tr\left[\frac{\idd + (-1)^a\sigma_x}{2}\left(U^\dagger N_4^bU\right)\right],
\end{align*}
where Eq.~\eqref{tomcom} was used to derive the second line. By using the standard procedures this self-testing statement can be made robust. For an equivalent result, obtained by considering the prepare-and-measure scenario without ancestral non-locality parent, see~\cite{TSV+20,DPV24,DPTV24}.

\subsection{Self-testing in a symmetric prepare-and-measure scenario}

Having established a correspondence between Bell and prepare-and-measure scenarios, it is natural to ask how a generalization of this map transforms another well-studied variation of the Bell scenario: a quantum network. Quantum networks extend the Bell scenario by involving at least three observers who receive classical inputs and produce classical outputs. A key feature distinguishing quantum networks is the presence of at least two mutually independent sources. The simplest example is the bilocality network, which comprises three parties connected by two independent sources as depicted in Fig.~\ref{fig:bilocal-network}; this is also a fundamental model for entanglement swapping.

 \begin{figure}[htbp]
    \centering

\begin{center}
\begin{tikzpicture}[
    every node/.style={align=center},    >=stealth
]
    \node[draw, rectangle, rounded corners, minimum width=1.5cm, minimum height=1cm] (Alice) {Alice};
    \node[draw, rectangle, rounded corners, minimum width=1.5cm, minimum height=1cm, right=1.5cm of Alice] (Bob) {Bob};
    \node[draw, rectangle, rounded corners, minimum width=1.5cm, minimum height=1cm, right=1.5cm of Bob] (Charlie) {Charlie};

    \draw[->] (Alice.north) ++(0,0.5cm) node[above] {$x$} -- (Alice.north);
    \draw[->] (Bob.north) ++(0,0.5cm) node[above] {$y$} -- (Bob.north);
    \draw[->] (Charlie.north) ++(0,0.5cm) node[above] {$z$} -- (Charlie.north);

    \draw[->] (Alice.south) -- ++(0,-0.5cm) node[below] {$a$};
    \draw[->] (Bob.south) -- ++(0,-0.5cm) node[below] {$b$};
    \draw[->] (Charlie.south) -- ++(0,-0.5cm) node[below] {$c$};

    \draw (Alice.east) -- (Bob.west) node[midway] (X1) {\large $\times$};
    \draw (Bob.east) -- (Charlie.west) node[midway] (X2) {\large $\times$};

    \node at ($(Alice.north east)!0.5!(Bob.north west)+(0,0.5cm)$) (*) {\Large $*$};
    \draw[dotted] (*) -- (Alice.north east);
    \draw[dotted] (*) -- (Bob.north west);

    \node at ($(Bob.north east)!0.5!(Charlie.north west)+(0,0.5cm)$) (*) {\Large $*$};
    \draw[dotted] (*) -- (Bob.north east);
    \draw[dotted] (*) -- (Charlie.north west);
\end{tikzpicture}
\end{center}
    \caption{The bilocal network involves two mutually independent sources and three parties. The fact that the sources are mutually independent affects the form of achievable global classical and quantum correlations.}
    \label{fig:bilocal-network}
\end{figure}
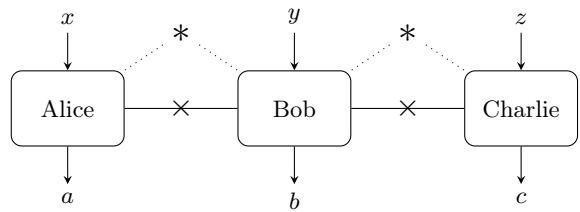

In the bilocality network, three quantum parties perform measurements $M_x^a$, $N_b^y$, and $P_c^z$, respectively. The global quantum state, represented as a tensor product $\varrho \otimes \varsigma$, describes the outputs of the two independent sources. The observed correlations among the parties follow the Born rule:
\begin{equation}
    p(a,b,c|x,y,z) = \tr\left[\left(M_x^a\otimes N_b^y\otimes P_c^z\right)\left(\varrho\otimes\varsigma\right)\right].
\end{equation}
A straightforward generalization of our map transforms the bilocality network into a symmetric prepare-and-measure scenario involving three parties. In this new framework, the two lateral parties prepare quantum states and send them to the central party, which performs measurements on the received states. Specifically, the map prescribes that both the input and output of the two lateral parties in the network scenario become inputs for the lateral parties in the prepare-and-measure scenario.

 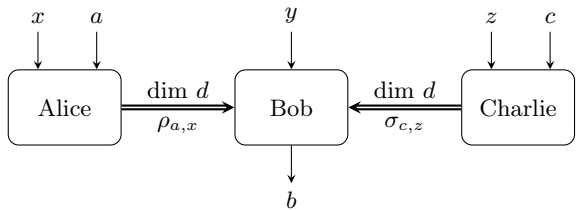
\begin{figure}[htbp]
    \centering
    \begin{tikzpicture}[
    every node/.style={align=center}, >=stealth
]
    \node[draw, rectangle, rounded corners, minimum width=1.5cm, minimum height=1cm] (Alice) {Alice};
    \node[draw, rectangle, rounded corners, minimum width=1.5cm, minimum height=1cm, right=1.5cm of Alice] (Bob) {Bob};
    \node[draw, rectangle, rounded corners, minimum width=1.5cm, minimum height=1cm, right=1.5cm of Bob] (Charlie) {Charlie};

    \draw[->] (Alice.125) ++(0,0.5cm) node[above] {$x$} -- (Alice.125);
    \draw[->] (Alice.50) ++(0,0.5cm) node[above] {$a$} -- (Alice.50);

    \draw[->] (Bob.north) ++(0,0.5cm) node[above] {$y$} -- (Bob.north);

    \draw[->] (Charlie.125) ++(0,0.5cm) node[above] {$z$} -- (Charlie.125);
    \draw[->] (Charlie.50) ++(0,0.5cm) node[above] {$c$} -- (Charlie.50);
    
    \draw[->] (Bob.south) -- ++(0,-0.5cm) node[below] {$b$};

  \draw[double,thick, ->] (Alice.east) -- (Bob.west) node[midway, above] {dim $d$}
  node[midway, below] {$\rho_{a,x}$};
    \draw[double, thick, <-] (Bob.east) -- (Charlie.west) node[midway, above] {dim $d$}
    node[midway, below] {$\sigma_{c,z}$};
\end{tikzpicture}
    \caption{Prepare and measure version of the bilocal network. The channels are dimensionally bounded.}
    \label{fig:PM-bilocal}
\end{figure}

Consequently, the lateral parties prepare states $\rho_{a,x}$ and $\sigma_{c,z}$, while the central party performs the measurement $N_b^y$, yielding the conditional probabilities:
\begin{equation}
    P(b|a,x,y,c,z) = \tr\left[N_b^y (\rho_{a,x}\otimes\sigma_{c,z})\right].
\end{equation}
If the pairs $(a,x)$ and $(c,z)$ are sampled uniformly at random, by using the Bayes rule we obtain:
\begin{equation}
    P(a,b,c|x,y,z) = \frac{1}{|\mathcal{A}\cdot\mathcal{C}|}\tr\left[N_b^y\left(\rho_{a,x}\otimes\sigma_{c,z}\right)\right],
\end{equation}
where $|\mathcal{A}|$ and $|\mathcal{C}|$ denote the cardinalities of the sets to which $a$ and $c$ belong. These prepare-and-measure correlations can be mapped back to their network counterparts using the product swap trick:
\begin{align}
    &\tr\left[B\cdot (A\otimes C)\right] = \\
    &d^2 \left(\bra{\phi^+_{AB_1}}\otimes\bra{\phi^+_{B_2C}}\right) 
    \left[ A^T \otimes B \otimes C^T \right]
    \left( \ket{\phi^+_{AB_1}}\otimes\ket{\phi^+_{B_2C}} \right)\nonumber
\end{align}
where $d$ is the dimension of the shared states, and $\ket{\phi^+}$ denotes a maximally entangled state.

Using this identity, we observe that if the bilocal correlations self-test a tensor product of two Bell states and the measurement operators of the three parties, the corresponding probabilities in the symmetric prepare-and-measure scenario also self-test the central party's measurement. Moreover, the prepared states of the lateral parties in this scenario correspond to the two remotely prepared states in the bilocal network. This mapping provides a novel approach for self-testing the Bell state measurement performed by the central party by leveraging existing bilocal self-testing results, such as those in~\cite{RenouPRL,vsupic2022genuine}.

In this example we chose to keep Bob as a measurement, and to introduce channels from Alice and Charlie to Bob.
Nothing prevents us from changing this; starting from the same Bell inequality and applying our protocol we could make Bob prepare 4 states and send them to Alice and Charlie, or make a sequential scenario in which Alice prepares a state, Bob transforms it and Charlie measures it.
In the future it might be interesting to study the interplay between all of these structures.

\subsection{Relaxing entanglement in network-device-independent protocols}

In this section, we briefly discuss how oup mapping could be applied to relax assumptions in network-device-independent protocols. Such protocols aim to transition from device-dependent setups to device-independent ones by embedding the device-dependent setup into a quantum network. The key assumption in these networks is the mutual independence of the sources, while the devices remain uncharacterised.

The process of converting a device-dependent protocol into a device-independent one can be summarized as follows. Consider without loss of generalisation a device-dependent setup where the involved parties perform characterized measurements to witness entanglement or to verify the specification of the source. To achieve device independence, one leverages the ability to perform a characterized measurement remotely by introducing an auxiliary party that shares a maximally entangled Bell pair with the original party. In this modified setup, the original party performs a Bell state measurement on their half of the Bell pair and their share of the uncharacterised state. Meanwhile, the auxiliary party performs the characterized measurement. Through the Bell state measurement, the uncharacterised state is teleported to the auxiliary party, where the characterized measurement is then performed directly on the teleported state. A graphical representation of this can be found in Fig. \ref{fig:network}.

\begin{figure*}
    \centering
    \begin{tikzpicture}[scale=0.8,  >=stealth]

  \def\L{2.5} 
  \def\h{0.7*\L} 
    
    \def\Lbig{7} 
  \def\hbig{0.7*\Lbig} 

      \def\Lm{4.8} 
  \def\hm{0.7*\Lm}
    
  \coordinate (A) at (-\L/2, \h/2);
  \coordinate (B) at (\L/2, \h/2);
  \coordinate (C) at (0, -\h);
  \coordinate (O) at (0,-\h/4);           

\coordinate (A1) at (-\Lbig/2, \hbig/2); 
  \coordinate (B1) at (\Lbig/2, \hbig/2);
  \coordinate (C1) at (0, -\hbig*5/6);

  \coordinate (Am) at (-\Lm/2, \hm/2); 
  \coordinate (Bm) at (\Lm/2, \hm/2);
  \coordinate (Cm) at (0, -\hm*6/7);

   \node [draw, rectangle, rounded corners, minimum width=0.8cm, minimum height=0.8cm, anchor=center] (Alice) at (A) {$A$};

    \node [draw, rectangle, rounded corners, minimum width=0.8cm, minimum height=0.8cm, anchor=center] (Bob) at (B) {$B$};
    
    \node [draw, rectangle, rounded corners, minimum width=0.8cm, minimum height=0.8cm, anchor=center] (Charlie) at (C) {$C$};

    \node [draw, rectangle, rounded corners, minimum width=0.8cm, minimum height=0.8cm, anchor=center] (Alice1) at (A1) {$A_1$};

    \node [draw, rectangle, rounded corners, minimum width=0.8cm, minimum height=0.8cm, anchor=center] (Bob1) at (B1) {$B_1$};
    
    \node [draw, rectangle, rounded corners, minimum width=0.8cm, minimum height=0.8cm, anchor=center] (Charlie1) at (C1) {$C_1$};

    \node[anchor=center] (Cen) at (O) {\Large{*}};

    \node[anchor=center] (Amm) at (Am) {\Large{*}};
    \node[anchor=center] (Bmm) at (Bm) {\Large{*}};
    \node[anchor=center] (Cmm) at (Cm) {\Large{*}};

    \draw[thick] (Charlie.south) -- (Cmm.north);
    \draw[thick] (Cmm.south) -- (Charlie1.north);
    \draw[thick] (Bob.40) -- (Bmm.200);
    \draw[thick] (Bmm.50) -- (Bob1.200);
    \draw[thick] (Alice.140) -- (Amm.340);
    \draw[thick] (Amm.130) -- (Alice1.340);
    
    \draw[thick] (Charlie.north) -- (Cen);
    \draw[thick] (Cen.30) -- (Bob.230);
    \draw[thick] (Cen.150) -- (Alice.310);
  
  \end{tikzpicture}
    \hspace{50pt}
    \begin{tikzpicture}[scale=0.8,  >=stealth]

  \def\L{2.5} 
  \def\h{0.7*\L} 
    
    \def\Lbig{7} 
  \def\hbig{0.7*\Lbig} 

      \def\Lm{4.8} 
  \def\hm{0.7*\Lm}
    
  \coordinate (A) at (-\L/2, \h/2);
  \coordinate (B) at (\L/2, \h/2);
  \coordinate (C) at (0, -\h);
  \coordinate (O) at (0,-\h/4);           

\coordinate (A1) at (-\Lbig/2, \hbig/2); 
  \coordinate (B1) at (\Lbig/2, \hbig/2);
  \coordinate (C1) at (0, -\hbig*5/6);

  \coordinate (Am) at (-\Lm/2, \hm/2); 
  \coordinate (Bm) at (\Lm/2, \hm/2);
  \coordinate (Cm) at (0, -\hm*6/7);

   \node [draw, rectangle, rounded corners, minimum width=0.8cm, minimum height=0.8cm, anchor=center] (Alice) at (A) {$A$};

    \node [draw, rectangle, rounded corners, minimum width=0.8cm, minimum height=0.8cm, anchor=center] (Bob) at (B) {$B$};
    
    \node [draw, rectangle, rounded corners, minimum width=0.8cm, minimum height=0.8cm, anchor=center] (Charlie) at (C) {$C$};

    \node [draw, rectangle, rounded corners, minimum width=0.8cm, minimum height=0.8cm, anchor=center] (Alice1) at (A1) {$A_1$};

    \node [draw, rectangle, rounded corners, minimum width=0.8cm, minimum height=0.8cm, anchor=center] (Bob1) at (B1) {$B_1$};
    
    \node [draw, rectangle, rounded corners, minimum width=0.8cm, minimum height=0.8cm, anchor=center] (Charlie1) at (C1) {$C_1$};

    \node[anchor=center] (Cen) at (O) {\Large{*}};

    \draw[double,thick, <-] (Charlie.south) -- (Charlie1.north) node[midway, right] {dim $d$};
    \draw[double,thick, <-] (Bob.40) -- (Bob1.200) node[midway, above=0.3cm] {dim $d$};
    \draw[double,thick, <-] (Alice.140)-- (Alice1.340) node[midway, above=0.3cm] {dim $d$};
    
    \draw[thick] (Charlie.north) -- (Cen);
    \draw[thick] (Cen.30) -- (Bob.230);
    \draw[thick] (Cen.150) -- (Alice.310);

  
  \end{tikzpicture}
    \caption{On the left, a graphical representation of the network-device-independent protocol described above. $A$, $B$ and $C$ are the original parties; $A_1$, $B_1$ and $C_1$ are the auxiliary parties.
    On the right, a simpler network-device-independent protocol obtained through our mapping.}
    \label{fig:network}
\end{figure*}
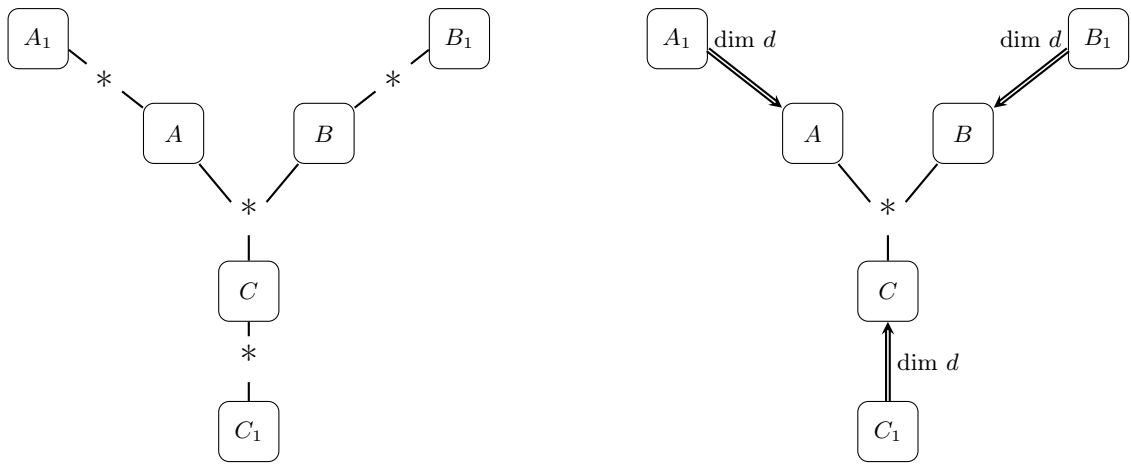

To complete the transition to a device-independent protocol, all measurements are made uncharacterised, and self-testing is employed to verify the Bell pair shared between the original and auxiliary parties. For instance, this can be achieved by demonstrating the maximal violation of the CHSH inequality. Moreover, using a chain of three CHSH inequalities, one can self-test a tomographically complete set of qubit measurements. This enables the self-testing of any arbitrary measurement performed by the auxiliary party. 

This procedure can be applied to every party involved in an $N$-partite device-dependent protocol, resulting in a $2N$-partite device-independent protocol. However, the method requires $N$ auxiliary parties, each sharing a Bell pair with one of the original parties. While effective, this approach necessitates significant resources, specifically $N$ Bell pairs. This method was first used to compose a protocol for universal device-independent entanglement certification~\cite{bowles2018device}, self-testing of all pure entangled states~\cite{vsupic2023quantum} and network-device-independent witnessing of indefinite causal order~\cite{dourdent2024network}.

Our mapping offers an alternative to this resource-intensive setup. Instead of requiring $N$ Bell pairs, the auxiliary parties can directly prepare and send the states that would have been remotely prepared for the original parties in the device-dependent protocol. This is possible under the assumption that the communication channels connecting the auxiliary parties to the original parties have a characterized dimension. By replacing the Bell-pair requirement with state preparation in a dimensionally characterized channel, our approach simplifies the requirements for implementing network-device-independent protocols while maintaining the desired security and functionality.

\section{Discussion and future directions}\label{sec:concl}
In this work, we have explored the connection between quantum correlations in device-independent non-local scenarios and semi-device-independent prepare-and-measure scenarios. Our approach establishes a systematic framework for mapping Bell inequalities natively formulated for the non-local scenario onto a prepare-and-measure structure, preserving essential quantum properties such as the quantum bound and self-testing features. 
The key tools are sum-of-squares decompositions and a property of maximally entangled states which we call the swap trick.

This translation offers alternative ways for understanding the fundamental limits of quantum communication tasks that do not rely on entanglement. By identifying the mathematical equivalences between these scenarios, we have shown that quantum violations of certain Bell inequalities naturally manifest as constraints on dimension-constrained communication in prepare-and-measure settings.
Interestingly, many of the translated inequalities that we found can be mapped by relabelling to known quantum prepare-and-measure games, for which self-testing proofs with different mathematical techniques were developed; our framework allows to unify these scattered proofs under a common formalism.

At a deeper conceptual level, our work contributes to the ongoing comparison between parallel or spatial scenarios, of which Bell scenario is a representative, and sequential or temporal scenarios, which in our case is the prepare-and-measure setting.
The characterization of quantum temporal correlations and their relation to spatial quantum correlations received significant attention~\cite{BTCV04,Fri10,FAB+11,CRG+18,CE24}. Notably, in \cite{Fri10} the author shows that \emph{although the set of joint probabilities realizable by spatial quantum correlations is strictly contained in the set of joint probabilities realizable by temporal quantum correlations, the set of realizable correlators is the same in the temporal case as in the spatial case}. This finding  closely parallels our findings: in section~\ref{sec:counterexample} we show that the set of prepare-and-measure quantum correlations is a strict superset of the set of quantum spatial correlations, however correlation Bell inequalities (those not containing marginal terms) are always maximally violated by maximally entangled states, implying that in our case as well quantum correlators realizable in the prepare-and-measure scenario are also realizable in the Bell scenario.
In the future we plan to explore the connection between our findings and those in~\cite{Fri10}.

Many other directions are also to be explored.
First of all, we showed that a sufficient condition for the quantum bound to be preserved in dimension $2$ is that the original Bell inequality is maximally violated by a maximally entangled state. A similar equivalence has been found in~\cite{EPQU25} in the context of characterising measurement incompatibility. Whether an alternative mapping can be found for other Bell inequalities remains an open question. 

Furthermore, it is relevant to investigate what happens when considering more than two parties: can we have an equivalent map from Bell multipartite to prepare-transform-measure scenarios?
This seems to be closely related to sequential quantum random access codes introduced in~\cite{Mohan_2019}. 

\section*{Acknowledgments}
MB acknowledges funding from QuantEdu France, a state aid managed by the French National Research Agency for France 2030 with the reference ANR-22-CMAS-0001.
I\v{S} and DM acknowledge funding from the PEPR integrated project EPiQ ANR-22-PETQ-0007 part of Plan France 2030.
ED acknowledges funding from the European Union’s Horizon Europe research and innovation program under the grant agreement No 101114043 (QSNP).
The authors thank Jef Pauwels, Marco Túlio Quintino, Máté Farkas and Victoria Wright for insightful discussions.

\newpage
\bibliographystyle{apsrev4-1}
\bibliography{biblio}

\newpage
\onecolumngrid

\appendix

\section{Qubit strategies}

In this appendix we report the explicit SOS decompositions of the bipartite inequalities that are mentioned in the main text.
In particular, for the Elegant Bell Inequality we also show explicitly the robustness bounds.
Similar techniques can be extend to the other inequalities as well.

\subsection{A family of Bell inequalities with binary inputs
and outputs}\label{app:family_binary}
Consider the family of Bell inequalities with binary inputs
and outputs introduced in~\cite{Le_2023} and in~\cite{barizien2023custom}. Recall that the Bell operator is parametrized by $\alpha$, $\beta$ and $\gamma$ :
\begin{align*}
    \mathcal{B}_{\alpha\beta\gamma}^{\text{NL}}(A_x,B_y)
    &= \cos(\alpha+\beta)\cos(\alpha + \gamma)A_0\otimes (\cos(\gamma)B_0-\cos(\beta)B_1) \\
    &+ \cos(\beta)\cos(\gamma)A_1\otimes (-\cos(\alpha+\gamma)B_0 +\cos(\alpha+\beta)B_1).
\end{align*}
This family includes the CHSH inequality for $\alpha=\gamma=0$ and $\beta=\pi$,
and many more interesting cases.
In~\cite{wooltorton2023deviceindependent} they showed that, under some conditions, reaching the quantum bound
\begin{equation*}
    b_q = \pm \sin(\alpha)\sin(\gamma-\beta)\sin(\alpha+\beta+\gamma)
\end{equation*}
self-tests the maximally entangled pair of qubits. 
They show this providing an SOS decomposition of the shifted Bell operator, which is the following
\begin{align}\label{eq:SOS_family}
    b_q\idd - \mathcal{B}_{\alpha\beta\gamma}^{\text{NL}}(A_x,B_y) = c_1P_1^\dagger P_1 + c_2P_2^\dagger P_2,
\end{align}
with coefficients
\begin{equation*} 
c_1 = -\frac{\cos(\gamma)\cos(\alpha+\gamma)}{2\sin(\alpha)},\quad c_2 = -\frac{\cos(\beta)\cos(\alpha+\beta)}{2\sin(\alpha)},
\end{equation*}
and polynomials
\begin{align*}
    P_1 &= \sin(\alpha)B_0 + \cos(\alpha+\beta)A_0 -\cos(\beta)A_1,\\
    P_2 &= \sin(\alpha)B_1 + \cos(\alpha+\gamma)A_0 -\cos(\gamma)A_1.
\end{align*}

Then, it is possible to prove that the unique optimal quantum strategy is achieved with a bipartite maximally entangled state, and projective observables which are unitarily equivalent to 
\begin{align*}
    & A_0 = \sigma_x, &B_0= \sin(\beta) \sigma_x + \cos(\beta) \sigma_y,\\
    & A_1 = \cos(\alpha) \sigma_x + \sin(\alpha) \sigma_y, &B_1=\sin(\gamma) \sigma_x + \cos(\gamma) \sigma_y.
\end{align*}
The formal proof of this statement can be found in the Appendix A of \cite{wooltorton2023deviceindependent}.
\subsubsection{Prepare-and-measure inequality}
In the prepare-and-measure game Alice will prepare four states $\rho_{a,x}$, that we can group as two quasi-observables
\begin{align*}
    \rho_0 = \rho_{0,0} - \rho_{1,0}, \qquad \rho_1 = \rho_{0,1} - \rho_{1,1}.
\end{align*}
We showed in the main text that these objects satisfy the same properties of observables, they are hermitian and their square is bounded by identity.
The prepare-and-measure Bell-like operator assumes the following form
\begin{align*}
    \mathcal{B}_{\alpha\beta\gamma}^{\text{PM}}(\tilde{\rho}_x,B_y)  
    = \frac{1}{2} &\Big[\cos(\alpha+\beta)\cos(\alpha + \gamma)\tilde{\rho}_0(\cos(\gamma)B_0-\cos(\beta)B_1) \\
    &+ \cos(\beta)\cos(\gamma)\tilde{\rho}_1(-\cos(\alpha+\gamma)B_0 +\cos(\alpha+\beta)B_1)\Big].
\end{align*}
Using the swap trick we can relate the new operator to the original non-local Bell operator:
\begin{align*}
    \tr\big[\mathcal{B}_{\alpha\beta\gamma}^{\text{PM}}(\tilde{\rho}_x,B_y)\big] = \bra{\phi_+} \mathcal{B}_{\alpha\beta\gamma}^{\text{NL}}(\tilde{\rho}_x,B_y)\ket{\phi_+}
\end{align*}
and consequently we inherit the SOS decomposition for the prepare and measure shifted Bell operator
\begin{align*}
    \tr\big[b_q\idd-\mathcal{B}_{\alpha\beta\gamma}^{\text{PM}}(\tilde{\rho}_x,B_y)\big] = \bra{\phi_+}  b_q\idd-\mathcal{B}_{\alpha\beta\gamma}^{\text{NL}}(\tilde{\rho}_x,B_y)\ket{\phi_+}
    \leq \sum_{\lambda = 1}^2 c_\lambda \bra{\phi_+}\tilde{P}_\lambda^2 \ket{\phi_+}  
\end{align*}
where the coefficients are the same as before, and the polynomials are expressed in terms of the quasi-observables
\begin{align*}
    \tilde{P}_1 &= \sin(\alpha)B_0 + \cos(\alpha+\beta)\tilde{\rho}_0 -\cos(\beta)\tilde{\rho}_1,\\
    \tilde{P}_2 &= \sin(\alpha)B_1 + \cos(\alpha+\gamma)\tilde{\rho}_0 -\cos(\gamma)\tilde{\rho}_1.
\end{align*}
We can then apply exactly the same reasoning of \cite{wooltorton2023deviceindependent} to retrieve the self-testing statement.

\subsection{Elegant Bell inequality}\label{app:EBI}

We illustrate the protocol by applying it to the elegant Bell inequality introduced in \cite{Gis07}.
First, we reproduce the known proof of the bound and the self-testing of the optimal strategy in the non-local scenario, using an appropriate SOS decomposition. We then demonstrate the robustness of the statement, showing that a near-optimal violation implies a strategy close to the optimal one. Finally, we translate these results to the prepare-and-measure version of the inequality.

\subsubsection{SOS decomposition}

The Bell operator corresponding to the elegant Bell inequality is defined as
\begin{align*}
    \mathcal{B}_{el}^{\text{NL}}\left(A_x,B_y\right) &=
    A_1 \otimes (B_1 + B_2 - B_3 - B_4) 
    + A_2 \otimes (B_1 - B_2 + B_3 - B_4) 
    + A_3 \otimes (B_1 - B_2 - B_3 + B_4) \\
    &= (A_1 + A_2 + A_3) \otimes B_1 
    + (A_1 - A_2 - A_3) \otimes B_2 
    + (-A_1 + A_2 - A_3) \otimes B_3
    + (-A_1 - A_2 + A_3) \otimes B_4.
\end{align*}

Optimal classical strategy achieves a value of $\beta_L = 6$, while quantum bound takes value $\beta_Q = 4 \sqrt{3}$ with the following strategy
\begin{align*}
    \ket{\psi} &= \frac{\ket{00}+\ket{11}}{\sqrt{2}}\\       A_1 = \sigma_X, \qquad   A_2 &= \sigma_Y,\qquad  A_3 = \sigma_Z, \\ B_1 = \frac{\sigma_X - \sigma_Y +\sigma_Z}{\sqrt{3}}\qquad B_2 = \frac{\sigma_X + \sigma_Y -\sigma_Z}{\sqrt{3}} &\qquad B_3 = \frac{-\sigma_X - \sigma_Y -\sigma_Z}{\sqrt{3}} \qquad B_4 = \frac{-\sigma_X + \sigma_Y +\sigma_Z}{\sqrt{3}} 
\end{align*}

In fact, this optimal strategy is unique (up to local isometries and complex conjugations).
More precisely, it can be shown that achieving the maximal violation of the elegant Bell inequality fixes all the anti-commutators of the measurement observables.
Let us consider the standard SOS decomposition for the shifted elegant Bell operator from \cite{Acin_2016}
\begin{align}\label{eq:SOS_EBI}
    \sum_{\lambda = 1}^4 P_\lambda^2 \leq 4 \sqrt{3} \mathds{1} - \mathcal{B}_{el}^{\text{NL}}\left(A_x,B_y\right)
\end{align}
with the following definition of the polynomials
\begin{align*}
    P_1 = \frac{+A_1+A_2+A_3}{\sqrt{3}} \otimes \idd - \idd \otimes B_1, &\qquad P_2 = \frac{+A_1-A_2-A_3}{\sqrt{3}} \otimes \idd - \idd \otimes B_2\\
    P_3 = \frac{-A_1+A_2-A_3}{\sqrt{3}} \otimes \idd - \idd \otimes B_3, &\qquad P_4 = \frac{-A_1-A_2+A_3}{\sqrt{3}} \otimes \idd - \idd \otimes B_4\\
\end{align*}

To simplify the notation, we will omit to explicitly write the tensor product with identity, and label the following sums of Alice's observables
\begin{align*}
    &\mathscr{A}_1 = \frac{+A_1+A_2+A_3}{\sqrt{3}} \qquad
    &\mathscr{A}_3 = \frac{-A_1+A_2-A_3}{\sqrt{3}} \\
    &\mathscr{A}_2 = \frac{+A_1-A_2-A_3}{\sqrt{3}} \qquad
    &\mathscr{A}_4 = \frac{-A_1-A_2+A_3}{\sqrt{3}} 
\end{align*}
such that the polynomials can be written as
\begin{equation}\label{eq:SOS_polynomial}
    P_\lambda = \mathscr{A}_\lambda - B_\lambda
\end{equation}

Taking the expectation value of Eq. \ref{eq:SOS_EBI}, we get
\begin{align*}
    \bra{\psi}\mathcal{B}_{el}^{\text{NL}}\left(A_x,B_y\right) \ket{\psi} \leq 4 \sqrt{3} - \sum_{\lambda = 1}^4 \bra{\psi}P_\lambda^2 \ket{\psi} \leq 4\sqrt{3}
\end{align*}
and the bound is saturated when all measurements are projective and squares in every sum equal to zero:
\begin{equation*}
    \sum_{\lambda = 1}^4 \bra{\psi}P_\lambda^2 \ket{\psi} = 0 \implies \bra{\psi}P_\lambda^2 \ket{\psi} = 0 \implies \| P_\lambda \ket{\psi} \| = 0 \qquad \forall \lambda.
\end{equation*}

Consider the following algebraic identities :
\begin{align}
    &\frac{4}{3} \{A_2,A_3\} = B_1^2 + B_2^2 - B_3^2 - B_4^2 - (B_1^2 - \mathscr{A}_1^2)- (B_2^2 - \mathscr{A}_2^2)+ (B_3^2 - \mathscr{A}_3^2)+ (B_4^2 - \mathscr{A}_4^2) \label{eq:EBI_id1} \\ 
    &\frac{4}{3} \{A_1,A_2\} = B_1^2 - B_2^2 - B_3^2 + B_4^2 - (B_1^2 - \mathscr{A}_1^2)+ (B_2^2 - \mathscr{A}_2^2)+ (B_3^2 - \mathscr{A}_3^2)- (B_4^2 - \mathscr{A}_4^2) \label{eq:EBI_id2} \\ 
    &\frac{4}{3} \{A_3,A_1\} = B_1^2 - B_2^2 + B_3^2 - B_4^2 - (B_1^2 - \mathscr{A}_1^2)+ (B_2^2 - \mathscr{A}_2^2)- (B_3^2 - \mathscr{A}_3^2)+ (B_4^2 - \mathscr{A}_4^2) \label{eq:EBI_id3}
\end{align}

The first vanishes if we consider projective observables; the second part also vanishes when evaluated on the state $\ket{\psi}$, because we can reformulate it as
\begin{equation*}
    (B_\lambda^2 - \mathscr{A}_\lambda^2) = (B_\lambda + \mathscr{A}_\lambda)(B_\lambda - \mathscr{A}_\lambda) = - (B_\lambda + \mathscr{A}_\lambda) P_\lambda.
\end{equation*}

Hence, we proved that a strategy that achieves the maximal bound $\beta_Q = 4 \sqrt{3}$ necessarily satisfy the following anti-commutation relations
\begin{align*}
    \{A_i,A_j\} \ket{\psi} = 0 \qquad \forall i \neq j \in [3]
\end{align*}
With very similar arguments we can prove that 
\begin{align*}
    \{B_i,B_j\} \ket{\psi} = \frac{1}{3} \idd \ket{\psi} \qquad \forall i \neq j \in [4]
\end{align*}

\subsubsection{Robust self-testing}

If we are $\epsilon$-close to maximally violate the inequality, we would like to infer that the strategy is $O(\sqrt{\epsilon})$-close to the projective observables for Alice and Bob that we showed in the previous section. In particular, we want to bound the anti-commutators of the operators. Let us start by rewriting an explicit version of the SOS decomposition without assuming projective observables :
\begin{equation*}
 \sum_\lambda P_\lambda^\dagger P_\lambda = 8 \idd - \frac{2}{\sqrt{3}} \mathcal{B}_{el}^{\text{NL}} + \sum_{i=1}^4 (-\idd + B_i^2) +\frac{4}{3} \sum_{j=1}^3 (-\idd + A_j^2)
\end{equation*}

Taking the expectation value over the state $\ket{\psi}$, and assuming that $\bra{\psi} \mathcal{B}_{el}^{\text{NL}} \ket{\psi} = 4\sqrt{3} - \epsilon$, we get the following 
\begin{equation*}
 0 \leq \sum_\lambda \bra{\psi} P_\lambda^\dagger P_\lambda \ket{\psi} = \frac{2}{\sqrt{3}} \epsilon + \bra{\psi}\left( -8 \idd + \sum_{i=1}^4 B_i^2 +\frac{4}{3} \sum_{j=1}^3 A_j^2 \right) \ket{\psi}
\end{equation*}

Rearranging the terms, and remembering that $A_j^2,B_i^2 \leq \idd$, we obtain the following inequality
\begin{equation*}
  8 - \frac{2}{\sqrt{3}} \epsilon \leq \bra{\psi}\left( \sum_{i=1}^4 B_i^2 +\frac{4}{3} \sum_{j=1}^3 A_j^2 \right) \ket{\psi} \leq 8
\end{equation*}

Since all the terms $0 \leq \bra{\psi} B_i^2 \ket{\psi} \leq 1$ and $0 \leq \bra{\psi} A_j^2 \ket{\psi} \leq 1$ are positive and bounded by $1$, in order to satisfy the lower bound it must be
\begin{equation}\label{eq:epsilon-proj}
    1- \frac{2}{\sqrt{3}}\epsilon \leq \bra{\psi} B_i^2 \ket{\psi} \leq 1 \qquad \forall i, \qquad
    1- \frac{2}{\sqrt{3}}\epsilon \leq \bra{\psi} A_j^2 \ket{\psi} \leq 1 \qquad \forall j
\end{equation}
meaning that, when we are close to the maximal violation of the inequality, the observables are close to be projective.

Now, the norm of the polynomials of the SOS polynomials is bounded by a function of $\epsilon$
\begin{equation*}
   \| P_\lambda \ket{\psi} \| = \| \mathscr{A}_\lambda - B_\lambda \ket{\psi}\| \leq \sqrt{\frac{2 \epsilon}{\sqrt{3}}}= \sqrt{\varepsilon} \qquad \forall \lambda
\end{equation*}

Since $A_i^2 \leq \idd$ and $B_i^2 \leq \idd$, the following norms are also bounded
\begin{align*}
   \| A_i  (\mathscr{A}_\lambda - B_\lambda)\ket{\psi}\| \leq \|(\mathscr{A}_\lambda - B_\lambda)\ket{\psi}\| \leq \sqrt{\varepsilon} \\
    \| B_i (\mathscr{A}_\lambda - B_\lambda)\ket{\psi}\| \leq \|(\mathscr{A}_\lambda -B_\lambda)\ket{\psi}\| \leq \sqrt{\varepsilon}
\end{align*}

Let us consider the identities \ref{eq:EBI_id1}-\ref{eq:EBI_id3}, and their norm :
\begin{align*}
    &\| \frac{4}{3} \{A_2,A_3\} \ket{\psi}\|= \|B_1^2 + B_2^2 - B_3^2 - B_4^2 - (B_1^2 - \mathcal{A}_1^2)- (B_2^2 - \mathcal{A}_2^2)+ (B_3^2 - \mathcal{A}_3^2)+ (B_4^2 - \mathcal{A}_4^2) \ket{\psi}\| \\
    &\| \frac{4}{3} \{A_1,A_2\}\ket{\psi}\|= \|B_1^2 - B_2^2 - B_3^2 + B_4^2 - (B_1^2 - \mathcal{A}_1^2)+ (B_2^2 - \mathcal{A}_2^2)+ (B_3^2 - \mathcal{A}_3^2)- (B_4^2 - \mathcal{A}_4^2)\ket{\psi}\| \\
    &\| \frac{4}{3} \{A_3,A_1\} \ket{\psi}\|= \| B_1^2 - B_2^2 + B_3^2 - B_4^2 - (B_1^2 - \mathcal{A}_1^2)+ (B_2^2 - \mathcal{A}_2^2)- (B_3^2 - \mathcal{A}_3^2)+ (B_4^2 - \mathcal{A}_4^2) \ket{\psi}\|
\end{align*}

Using the triangle inequality, we can cut every sum in two part.
As for the first part, we know that Bob's observables are close to be projective; more precisely, using Eq. \ref{eq:epsilon-proj} we can deduce that that :
\begin{align*}
    |\bra{\psi} ( B_i^2 + B_j^2 - B_k^2 - B_l^2) \ket{\psi}| \leq 2 \varepsilon
\end{align*}
Using Cauchy–Bunyakovsky–Schwarz inequality we find the following bound for the norm of the first addend
\begin{align*}
    &\| ( B_i^2 + B_j^2 - B_k^2 - B_l^2) \ket{\psi} \|^2 = \bra{\psi} ( B_i^2 + B_j^2 - B_k^2 - B_l^2)^2 \ket{\psi} \leq \bra{\psi} ( B_i^2 + B_j^2 - B_k^2 - B_l^2) \ket{\psi}^2 \leq 4 \varepsilon^2
    \\
    & \implies \| ( B_i^2 + B_j^2 - B_k^2 - B_l^2) \ket{\psi} \| \leq 2 \varepsilon \qquad \forall i,j,k,l \in [4].
\end{align*}

Focusing on the second part, for every $\lambda \in [4]$ we can bound it using again the triangle inequality
\begin{align*}
 \| \pm (B_\lambda^2 - \mathcal{A}_\lambda^2)  \ket{\psi} \|
    \leq \| \pm B_\lambda (B_\lambda - \mathcal{A}_\lambda) \ket{\psi} \| + \|\pm \mathcal{A}_\lambda (B_\lambda - \mathcal{A}_\lambda)  \ket{\psi} \|
    \leq (1+\sqrt{3})\sqrt{\varepsilon}.
\end{align*}

Putting everything back together, we can bound the norm of the anti-commutators
\begin{align*}
    &\| \frac{4}{3} \{A_i,A_j\}  \ket{\psi} \| \leq 2\varepsilon + 4(1+\sqrt{3})\sqrt{\varepsilon} \qquad \forall i \neq j \in [3].
\end{align*}

In conclusion, we proved that the anti-commutators of Alice-observable are $\epsilon$-close to the optimal case :
\begin{align*}
    &\| \{A_i,A_j\} \ket{\psi}\| \leq \frac{3}{2} \varepsilon + 3 (1+\sqrt{3})\sqrt{\varepsilon}
    \qquad \forall i \neq j \in [3]
\end{align*}

\subsubsection{Prepare-and-measure elegant Bell inequality}

In the associated prepare-and-measure scenario, Alice has $|\mathcal{X}||\mathcal{A}|=6$ preparations, Bob has $|\mathcal{Y}|=3$ dichotomic measurements, and they can communicate with a channel of dimension $2$.
It is beneficial to think of Alice's set of qubits as three pairs ($a \in [2]$) labelled by $x \in [3]$, that can be fully characterised by 6 angles $\omega_i$ and $\psi_{ij}$ (up to global rotations) :
\begin{align}\label{eq:EBI_quasiobv}
    & \tilde{\rho}_1 = \rho_{01} - \rho_{11} = \Vec{v}_1 \cdot \Vec{\sigma} & \norm{\Vec{v}_1} = \cos(\omega_1)\\
    & \tilde{\rho}_2 = \rho_{02} - \rho_{12} = \Vec{v}_2 \cdot \Vec{\sigma} & \norm{\Vec{v}_2} = \cos(\omega_2)\\
    & \tilde{\rho}_3 = \rho_{03} - \rho_{13} = \Vec{v}_3 \cdot \Vec{\sigma} & \norm{\Vec{v}_3} = \cos(\omega_3)\label{eq:EBI_quasiobv_fine}
\end{align}
\begin{align*}
    &\frac{1}{2} \{\tilde{\rho}_1,\tilde{\rho}_2\} = \Vec{v}_1 \cdot \Vec{v}_2 \mathds{1} = \cos(\omega_1) \cos(\omega_2) \cos(\psi_{12}) \mathds{1}\\
    &\frac{1}{2} \{\tilde{\rho}_2,\tilde{\rho}_3\} = \Vec{v}_2 \cdot \Vec{v}_3 \mathds{1} = \cos(\omega_2) \cos(\omega_3) \cos(\psi_{23}) \mathds{1}\\
    &\frac{1}{2} \{\tilde{\rho}_3,\tilde{\rho}_1\} = \Vec{v}_3 \cdot \Vec{v}_3 \mathds{1} = \cos(\omega_3) \cos(\omega_1) \cos(\psi_{31}) \mathds{1}
\end{align*}

The elegant Bell inequality for the prepare and measure scenario is the following :
\begin{align*}
    \text{tr}\left[\mathcal{B}_{el}^{\text{PM}}\left(\tilde{\rho}_x,B_y\right)\right]
    &= \frac{1}{2}
    \text{tr}\left[\left( \tilde{\rho}_1 \left(B_1 + B_2 - B_3 - B_4\right) + \tilde{\rho}_2 \left(B_1 - B_2 + B_3 - B_4\right) + \tilde{\rho}_3 \left(B_1 - B_2 - B_3 + B_4\right)\right)\right]
\end{align*}

Using the swap trick, we can turn traces into bra-ket products, and retrieve the form of the non-local inequality
\begin{align*}
    \text{tr}\left[\mathcal{B}_{el}^{\text{PM}}\left(\tilde{\rho}_x,B_y\right)\right] = 
    \bra{\phi^+} \mathcal{B}_{el}^{\text{NL}}\left(\tilde{\rho}_x,B_y\right) \ket{\phi^+}
\end{align*}
and its associated SOS decomposition
\begin{align*}
     \frac{\sqrt{3}}{2} \sum_{\lambda = 1}^4 \bra{\phi^+} \Tilde{P}_\lambda^2\left(\tilde{\rho}_x,B_y\right) \ket{\phi^+} \leq 4 \sqrt{3} - \bra{\phi^+} \mathcal{B}_{el}^{\text{NL}}\left(\tilde{\rho}_x,B_y\right) \ket{\phi^+}
\end{align*}
where the shared state is fixed to be the maximally entangled bipartite state, and instead of Alice's observables we have the quasi-observables defined in Eqs.~\ref{eq:EBI_quasiobv}-\ref{eq:EBI_quasiobv_fine}. Notice that these quasi-observables $\Tilde{\rho}_x$ satisfy the essential properties of an observable that we need for the SOS decomposition
\begin{align*}
    \text{tr}(\Tilde{\rho}_x) = 0 \qquad \Tilde{\rho}_x^2 = \|\Vec{v}_x\|^2 \idd \leq \idd
\end{align*}
and the last inequality is tight if and only if the states are pure (equivalent to the projectivity condition for observables).

The maximal violation of $\mathcal{B}_{el}^{\text{NL}}\left(\tilde{\rho}_x,B_y\right)$ is achieved with a maximally entangled bipartite state, and imposes the following constraints on Alice's quasi-observables
\begin{equation*}
    \tilde{\rho}_i^2 = \idd \qquad \forall i \in [3], \qquad
    \{\tilde{\rho}_i,\tilde{\rho}_j\}\ket{\phi_+} = 0 \qquad \forall i \neq j \in [3]
\end{equation*}
which is equivalent to
\begin{equation*}
    \omega_i = \pi \qquad \forall i \in [3], \qquad
    \psi_{ij} = \frac{\pi}{2} \qquad \forall i \neq j \in [3].
\end{equation*}
Hence, three quasi-observables have exactly the properties of three mutually anti-commuting qubit observables. Since, states Alice prepares were characterized up to global rotation, this implies that there exists a unitary matrix $U$ such that
\begin{equation}
    U\tilde{\rho}_1U^\dagger = \sigma_z, \quad U\tilde{\rho}_2U^\dagger = \sigma_x, \quad U\tilde{\rho}_3U^\dagger = \pm\sigma_y
\end{equation}
The same reasoning for the self-testing in the non-local case can be applied to further characterize Bob's measurements. 

The robustness is directly inherited from the non-local proof, with the difference that we do not need to consider a deviation between the real shared state and the maximally entangled state, since in our case the latter is just a mathematical tool.
To be more precise, for a violation of $\text{tr}\left[\mathcal{B}_{el}^{\text{PM}}\left(\tilde{\rho}_x,B_y\right)\right] = 4\sqrt{3} - \epsilon$, we can infer
\begin{align*}
    & 1- \varepsilon \leq \bra{\phi^+} \Tilde{\rho}_i^2\ket{\phi^+} \leq 1 
    && 1- \varepsilon \leq \cos(\omega_i) \leq 1
    &&\forall i\in [3]\\
    & \| \{ \Tilde{\rho}_i, \Tilde{\rho}_j\}\ket{\phi^+} \| \leq \frac{3}{2}\varepsilon + 3(1+\sqrt{3})\sqrt{\varepsilon} 
    && \cos(\psi_{ij}) \leq \frac{3}{2} (1+\sqrt{3})\sqrt{\varepsilon} + o(\varepsilon)
    && \forall i \neq j \in [3] 
\end{align*}
which corresponds to almost pure and pairwise orthogonal states ($\omega_i \approx \pi$), and almost anti-commuting quasi-observables ($\psi_{ij} \approx \pi/2$).

\subsection{The chained Bell inequality}\label{app:chained}

The chained Bell inequalities are another generalisation of CHSH, where we still consider two parties and dichotomic observables, but a larger number of measurement per party:
\begin{align}\label{CBI}
    \mathcal{B}^{ch(n)} = \sum_{i=1}^n (A_i \otimes B_i + A_{i+1}\otimes B_i) \qquad \text{where } A_{n+1} \equiv - A_1
\end{align}
The local and Tsirelson bounds are
\begin{align*}
        \beta_L=\max_{\mathcal{L}} \langle \mathcal{B}^{ch,n} \rangle = 2n-2, \qquad \beta_Q=\max_{\mathcal{Q}} \langle \mathcal{B}^{ch,n} \rangle = 2n \cos(\pi/2n)
\end{align*}
and the latter is achieved by using  maximally entangled bipartite state, and the following observables:
\begin{align*}
    &A_i = \sin(\phi_i)\sigma_x + \cos(\phi_i)\sigma_z,
    &&\phi_i=[(i-1)\pi]/n,\\
    &B_i = \sin(\phi'_i)\sigma_x + \cos(\phi'_i)\sigma_z,
    &&\phi_i'=[(2i-1)\pi]/2n,
\end{align*}
which are $n$ measurements per party maximally spread on one single plane of the Bloch sphere.

\subsubsection{Prepare-and-measure chained inequality}
To characterize the $2 n$ states Alice prepares, we need several angles: firstly, $n$ angles describing the possible non-orthogonality of the pairs
\begin{align*}
    \rho_x= \rho_{0x} -\rho_{1x} = \Vec{x}_i\cdot \Vec{\sigma} \qquad \norm{\Vec{x}_i} = \cos(\omega_x) \qquad \forall i \in [1,n]
\end{align*}
and $n$ angles describing the angles between the different \textit{basis} containing the pairs
\begin{align*}
    \frac{1}{2}\{\rho_i,\rho_j\} = \Vec{x}_i \cdot \Vec{x}_j \mathds{1} = \cos(\omega_i)\cos(\omega_j) \cos(\psi_{ij}) \mathds{1}
\end{align*}

The prepare-and-measure version of the Bell functional is
\begin{align}\label{CBI_PM}
    \text{tr}[\mathcal{B}^{ch,PM}] &=
    \frac{1}{2} \text{tr} \sum_i \Big[
    \rho_i B^T_i + \rho_{i+1} B^T_i
    \Big] \qquad \text{where } \rho_{n+1} \equiv -\rho_{1}\notag\\
    &=\bra{\phi_+}
    \sum_i \Big[ \rho_i \otimes B^T_i + \rho_{i+1} \otimes B^T_i\Big]
    \ket{\phi_+}
\end{align}

In \cite{SASA16} the authors prove the self testing property of the chained Bell inequalities, using the SOS decomposition of first and second order (first alone is not enough).
Since everything is formally the same, we can blindly translate the SOS decomposition and the self-testing statement.
From the first order SOS decomposition in \cite{SASA16} it follows that the functional is going to be maximised if and only if $\rho_i^2 = B_j^2 = \mathds{1}$, which means
\begin{align*}
    &\rho_i^2 = \mathds{1}  &\implies & \omega_i= \pi \qquad \forall i,\\
    &B_j^2 = \mathds{1}  &\implies & \text{projective measurements}\qquad \forall j.
\end{align*}
Then, to fix the anti-commutators as well we will need the second order of the SOS decomposition.

To do so we use the following results, which are proven in Appendix B of \cite{SASA16}, and are valid when we reach the maximal quantum violation of the chained Bell inequality:
\begin{itemize}
    \item n even:
    \begin{align*}
    &\{A_1, A_{\frac{n}{2}+1}\}=0\\
    &A_i = \sin(\phi_i)  A_{\frac{n}{2}+1} +\cos(\phi_i) A_1\\
    &B_i = \sin(\phi_i')  B_{\frac{n}{2},\frac{n}{2}+1} + \cos(\phi_i') B_{1,-n}
    \end{align*}
    \item n odd:
    \begin{align*}
    &\{A_1, A_{\frac{n+1}{2}}+ A_{\frac{n+3}{2}}\}=0\\
    &A_i = \sin(\phi_i)  A_{\frac{n+1}{2},\frac{n+3}{2}} + \cos(\phi_i) A_1\\
    &B_i = \sin(\phi'_i)  B_{\frac{n+1}{2}} + \cos(\phi'_i) B_{1,-n}    
    \end{align*}    
\end{itemize}
with the following definitions:
\begin{align*}
B_{i-1,i} \equiv \frac{B_i+B_{i-1}}{2 \cos(\frac{\pi}{2n})} \qquad
    A_{i-1,i} \equiv \frac{A_i+A_{i-1}}{2 \cos(\frac{\pi}{2n})}
\end{align*}

From this we can automatically check that, for both n odd or even
\begin{align*}
    \{A_i,A_j\} = (s_i s_j + c_i c_j) 2 \mathds{1} = \cos\Big((i-j)\frac{\pi}{n}\Big) 2\mathds{1} \implies \psi_{ij} = (i-j)\frac{\pi}{n}
\end{align*}
and for the $B$'s is completely symmetric.

\section{Qudit strategies}\label{app:qudits}

\subsection{Weaker condition for qudits}\label{app:qudits-weakercondition}
In this subsection we want to proof this very useful bound for the quasi-observables:
\begin{align}\label{eq:weaker-condition}
    \text{tr} \left( \sum_{n=1}^{d-1} \Tilde{\rho}^{(n)} \Tilde{\rho}^{(-n)}\right) \leq d (d-1)
\end{align}

By simply expanding the definition of the quasi-observable we obtain the following:
\begin{align}
    \text{tr} \left( \sum_{n=1}^{d-1} \Tilde{\rho}^{(n)} \Tilde{\rho}^{(-n)}\right) 
    &= \sum_{n=1}^{d-1} \sum_{a=0}^{d-1} \sum_{b=0}^{d-1} w^{n(a-b)}  \text{tr}(\rho_a \rho_b) =\\
    &= (d-1)  \sum_{a=0}^{d-1} \sum_{b=0}^{d-1} \delta_{a,b} \text{tr}(\rho_a \rho_b) =\\
    & = (d-1) \sum_{a=0}^{d-1} \text{tr}(\rho_a^2) \leq d (d-1)
\end{align}
where in the second line we used the definition of the Kronecker delta as an integer sum
\begin{equation}
    \sum_{k=1}^D w^{k n} = D \delta_{n,0}
\end{equation}
and in the last line we used the simple fact that $\text{tr}(\rho^2)\leq 1$ for every quantum state, where the equality is satisfied only for pure states.
Hence, the only way to saturate this bound is for all of these states to be pure.

\subsection{Generalised CHSH inequality}\label{app:qudit-genCHSH}
Recall the functional presented in \cite{KST+19}:
\begin{equation*}
    \mathcal{B}_d^{\text{KST+}}= \frac{1}{d^3} \sum_n \lambda_n \sum_{x,y} w^{nxy} A_x^{(n)} \otimes  B_y^{(n)}
\end{equation*}

We define the trace of the PM functional as:
\begin{align*}
    \text{tr}(\mathcal{B}_d^{\text{KST+},PM})&=\frac{1}{d} \text{tr} \left( \frac{1}{d^3} \sum_n \lambda_n \sum_{x,y} w^{nxy} \rho_x^{(n)} \cdot  M_y^{(n)} \right)=\\
    &=\bra{\phi^+}\frac{1}{d^3} \sum_n \lambda_n \sum_{x,y} w^{nxy} \rho_x^{(n)} \otimes  M_y^{(n)} \ket{\phi^+}=\\
    &=\bra{\phi^+}\frac{1}{d^3} \sum_n \sum_j \rho_j^{(n)} \otimes  C_j^{(n)} \ket{\phi^+}
\end{align*}
where we used the same notation of \cite{KST+19}
\begin{align*}
    C_j^{(n)} = \frac{\lambda_n}{\sqrt{d}} \sum_k w^{njk} M_k^{(n)}
\end{align*}
From this point, we can reproduce the same manipulations presented in \cite{KST+19}, where instead of the observable $A_x^{(n)}$ we have the quasi-observable $\Tilde{\rho}_x^{(n)}$.
There are only two technical points to overcome, that we briefly discuss.
First of all, it is not true anymore that $\Tilde{\rho}_j^{(0)} = \mathds{1}$, but the following bound is easy to verify:
\begin{equation*}
    \frac{1}{d^2} \sum_j \bra{\phi^+} \rho_j^{(0)} \otimes \mathds{1} \ket{\phi^+} = \frac{1}{d^2} \frac{1}{d} \sum_j \text{tr}(\rho_j^{(0)}) \leq  \frac{1}{d}
\end{equation*}
Secondly, it is not true that $\Tilde{\rho}_j^{(n)} \Tilde{\rho}_j^{(-n)} \leq \mathds{1}$, but we can use the weaker conditions for qubits that we proved above (Eq. \ref{eq:weaker-condition}):
\begin{align*}
    \frac{1}{d^2\sqrt{d}}\sum_{n=1}^{\frac{d-1}{2}}
    \bra{\phi^+}&\left( \sum_j \rho_j^{(n)} \rho_j^{(-n)} \otimes \mathds{1} + \mathds{1}\otimes \sum_k M_k^{(n)} M_k^{(-n)} - \sum_j [L_j^{(n)}]^\dagger L_j^{(n)} \right)\ket{\phi^+}\leq\\
    &\leq \frac{d-1}{2 d \sqrt{d}} + \frac{1}{d^3\sqrt{d}} \sum_j\text{tr}
    \left( \sum_{n=1}^{\frac{d-1}{2}} \rho_j^{(n)} \rho_j^{(-n)} \right) \\
    &\leq \frac{d-1}{ d \sqrt{d}}
     \end{align*}

\subsection{Generalised Chained inequality}\label{app:qudit-genChained}
The trace of the PM version of the functional presented in \cite{SAT+17} is:
\begin{align}
    \text{tr}(\mathcal{B}_d^{\text{SAT+},PM}) &= \frac{1}{d} \text{tr} \left(\sum_{i=1}^m\sum_{l=1}^{d-1} \Tilde{\rho}_i^{(l)} \cdot (a_l M_i^{(d-l)} + a_l^* M_{i-1}^{(d-l)})\right)=\notag\\
    &=\bra{\phi^+} \sum_{i=1}^m\sum_{l=1}^{d-1} \Tilde{\rho}_i^{(l)} \otimes (a_l M_i^{(d-l)} + a_l^* M_{i-1}^{(d-l)}) \ket{\phi^+}=\\
    &=\bra{\phi^+} \sum_{i=1}^m\sum_{l=1}^{d-1} \Tilde{\rho}_i^{(l)} \otimes \overline{M_i^l} \ket{\phi^+}
\end{align}
which is formally equivalent to its corresponding non-local version, calculated on the maximally entangled state.
We directly adapt the SOS decomposition:
\begin{align*}
    \beta_Q \mathds{1} - \mathcal{B}_d^{\text{SAT+},PM} = \frac{1}{2} \sum_{i=1}^m \sum_{k=1}^{d-1} P_{ik}^\dagger P_{ik} + \frac{1}{2} \sum_{i=1}^{m-2} \sum_{k=1}^{d-1} T_{ik}^\dagger T_{ik}
\end{align*}
with
\begin{align*}
    P_{ik}=\mathds{1} \otimes \overline{M_i^k} - (\Tilde{\rho}_i^{(k)})^\dagger \otimes \mathds{1} \qquad \text{and} \qquad T_{ik}=f(M_2^{d-k}, M_{i+2}^{d-k}, M_{i+3}^{d-k})
\end{align*}
We need to check if the quasi-observables $\Tilde{\rho}_i^{(n)}$ are still allowing us to retrieve the same bound for the observables $A_i^n$.

Once again, we need to use the weaker condition (Eq. \ref{eq:weaker-condition}) to recover the bound:
\begin{align*}
    \bra{\phi^+}  \sum_{i=1}^{m} \sum_{k=1}^{d-1} \rho_i^{(k)}\rho_i^{(-k)} \otimes \mathds{1} \ket{\phi^+} =
    \frac{1}{d} \sum_{i=1}^{m}\text{tr}\left( \sum_{k=1}^{d-1} \rho_i^{(k)}\rho_i^{(-k)}\right) \leq \frac{1}{d}m d (d-1) = m (d-1)
\end{align*}

\section{Connection to games}\label{app:games}
Let us start by motivating more explicitly the equivalence between CHSH and QRAC $2^2\to1$.
The most intuitive way to see the correspondence is to explicitly write the truth table for the winning conditions of the two games, reported in Tab. \ref{tab:truthtables}.

\begin{table}[!ht]
\centering
\begin{minipage}{0.33\linewidth}
\centering
\begin{equation*}
    x \cdot y = a \oplus b
\end{equation*}
\begin{tabular}{cccc|c}
 $x$ & $a$ & & $y$ & $b$ \\
\toprule
0 & 0 & & 0 & 0\\
0 & 0 & & 1 & 0\\
\midrule
0 & 1 & & 0 & 1\\
0 & 1 & & 1 & 1\\
\midrule
1 & 0 & & 0 & 0\\
1 & 0 & & 1 & 1\\
\midrule
1 & 1 & & 0 & 1\\
1 & 1 & & 1 & 0\\
\end{tabular}
\end{minipage}%
\begin{minipage}{0.33\linewidth}
\centering
\begin{equation*}
    g = x_y
\end{equation*}
\begin{tabular}{cccc|c}
 $x_0$ & $x_1$ & & $y$ & $g$ \\
\toprule
0 & 0 & & 0 & 0\\
0 & 0 & & 1 & 0\\
\midrule
0 & 1 & & 0 & 0\\
0 & 1 & & 1 & 1\\
\midrule
1 & 0 & & 0 & 1\\
1 & 0 & & 1 & 0\\
\midrule
1 & 1 & & 0 & 1\\
1 & 1 & & 1 & 1\\
\end{tabular}
\end{minipage}
\begin{minipage}{0.33\linewidth}
\centering
\begin{equation*}
    x \cdot y = a \oplus b \qquad \wedge \qquad b = x'_y
\end{equation*}
\begin{tabular}{ccccccc|c}
 $x$ & $a$& & $x_0'$ & $x_1'$ & & $y$ & $b$ \\
\toprule
0 & 0 & & 0 & 0 & & 0 & 0\\
0 & 0 & & 0 & 0 & & 1 & 0\\
\midrule
0 & 1 & & 1 & 1 & & 0 & 1\\
0 & 1 & & 1 & 1 & & 1 & 1\\
\midrule
1 & 0 & & 0 & 1 & & 0 & 0\\
1 & 0 & & 0 & 1 & & 1 & 1\\
\midrule
1 & 1 & & 1 & 0 & & 0 & 1\\
1 & 1 & & 1 & 0 & & 1 & 0\\
\end{tabular}
\end{minipage}
\caption{The truth table of the winning conditions of CHSH (on the left) and QRAC $2^2\to 1$ (in the center). On the right, a relabelling of CHSH which is equivalent to QRAC $2^2\to 1$.}
\label{tab:truthtables}
\end{table}

In both tables we identified $4$ blocks, corresponding to the $4$ different preparations of Alice. We see that, except for the first block corresponding to $\rho_{ax}=\rho_{x_0 x_1}=\rho_{00}$, the winning condition differs in the two cases.
However, the two becomes the same if we allow for a change of variable for the labels of the prepared state $(x,a)$ or $(x_0,x_1)$ (see last table of Tab. \ref{tab:truthtables}).
We can apply exactly the same reasoning to the elegant Bell inequality and the QRAC $3^2 \to 1$.

\end{document}